\documentclass[10pt,final, twocolumn]{IEEEtran}

\usepackage{tabularx}

\usepackage{algorithmic,algorithm}
\usepackage{ae}
\usepackage[T1]{fontenc}
\usepackage{epsf}
\usepackage{epsfig}
\usepackage{amsmath}
\usepackage{amsfonts}
\usepackage{amssymb}
\usepackage{amsxtra}
\usepackage{latexsym}
\usepackage{color}
\usepackage{subfigure}
\usepackage{graphics}
\usepackage{algorithm}
\usepackage{algorithmic}
\usepackage{mathrsfs}
\usepackage{cite}
\usepackage{url}
\usepackage{ulem}

\begin{document} 

\title{Group Frames with Few Distinct Inner Products and Low Coherence} 

\author{Matthew Thill \hspace{12pt} Babak Hassibi \\ 
Department of Electrical Engineering, Caltech, Pasadena, CA  
 }

\maketitle 

\newtheorem{definition}{Definition}
\newtheorem{theorem}{Theorem}
\newtheorem{lemma}{Lemma}
\newtheorem{corollary}{Corollary}

\newcommand{\stexp}{\mbox{$\mathbb{E}$}}    
\newcommand{\beq}{\begin{equation}}
\newcommand{\eeq}{\end{equation}}
\newcommand{\bea}{\begin{eqnarray}}
\newcommand{\eea}{\end{eqnarray}}
\def\G{{\mathcal{G}}}
\newcommand{\Prob}{\ensuremath{\mathbb{P}}}

\newcommand{\matr}[1]{\mathbf{#1}}
\newcommand{\vect}[1]{\mathbf{#1}}

\newcommand{\ve}{\vect{e}}
\newcommand{\hve}{\vect{\hat e}}
\newcommand{\vnu}{\vect{\nu}}

\newcommand{\norm}[2]{\lVert #1 \rVert_{#2}}
\newcommand{\setweightR}[2]{\Sigma_{\R^{#1}}^{(#2)}} 

\newcommand{\eg}{\emph{e.g.}}
\newcommand{\ie}{\emph{i.e.}}
\newcommand{\etal}{\emph{et al.}}
\newcommand{\confer}{\emph{cf.}}

\newcommand{\Q}{\mathbb{Q}} 
\newcommand{\R}{\mathbb{R}} 
\newcommand{\Z}{\mathbb{Z}} 
\newcommand{\C}{\mathbb{C}} 
\newcommand{\vv}{\vect{v}} 
\newcommand{\U}{\matr{U}} 
\newcommand{\M}{\matr{M}} 
\newcommand{\x}{\vect{x}}
\newcommand{\y}{\vect{y}} 
\newcommand{\cc}{\vect{c}} 
\newcommand{\va}{\vect{a}} 
\newcommand{\w}{\vect{w}} 
\newcommand{\Tr}[1]{\text{Tr}(#1)} 
\newcommand{\diag}{\operatorname{diag}}

\begin{abstract} 
Frame theory has been a popular subject in the design of structured signals and codes in recent years, with applications ranging from the design of measurement matrices in compressive sensing, to spherical codes for data compression and data transmission, to spacetime codes for MIMO communications, and to measurement operators in quantum sensing. High-performance codes usually arise from designing frames whose elements have mutually low coherence. Building off the original ``group frame'' design of Slepian which has since been elaborated in the works of Vale and Waldron, we present several new frame constructions based on cyclic and generalized dihedral groups. Slepian's original construction was based on the premise that group structure allows one to reduce the number of distinct inner pairwise inner products in a frame with $n$ elements from $\frac{n(n-1)}{2}$ to $n-1$. All of our constructions further utilize the group structure to produce tight frames with even fewer distinct inner product values between the frame elements. When $n$ is prime, for example, we use cyclic groups to construct $m$-dimensional frame vectors with at most $\frac{n-1}{m}$ distinct inner products. We use this behavior to bound the coherence of our frames via arguments based on the frame potential, and derive even tighter bounds from combinatorial and algebraic arguments using the group structure alone. In certain cases, we recover well-known Welch bound achieving frames. In cases where the Welch bound has not been achieved, and is not known to be achievable, we obtain frames with close to Welch bound performance. 
\end{abstract} 


\let\thefootnote\relax\footnotetext{Copyright (c) 2015 IEEE. Personal use of this material is permitted. However, permission to use this material for any other purposes must be obtained from the IEEE by sending a request to pubs-permissions@ieee.org.

Email: \{mthill,hassibi\}@caltech.edu. This work was supported in part by the National Science Foundation under grants CNS-0932428, CCF-1018927, CCF-1423663 and CCF-1409204, by a grant from Qualcomm Inc., by NASAÕs Jet Propulsion Laboratory through the President and DirectorÕs Fund, by King Abdulaziz University, and by King Abdullah University of Science and Technology.}

\section{Introduction: Frames and Coherence} 
\label{sec:matcohframe}

Recall that a frame is the following generalization for the basis of a vector space: 

\begin{definition} Let $\mathcal{V}$ be a vector space equipped with an inner product $\langle \cdot, \cdot \rangle$ (or more specifically, a separable Hilbert space). A set of elements $\{f_k\}_{k \in \mathcal{I}}$, where $\mathcal{I}$ is a countable index set, is a \textit{frame} for $\mathcal{V}$ if there exist positive constants $A$ and $B$ such that 
\begin{align} 
A || f ||_2^2 \le \sum_{k \in \mathcal{I}} | \langle f, f_k \rangle |^2 \le B || f ||_2^2, \label{eqn:framedef} 
\end{align} 
for all $f \in \mathcal{V}$. A frame is called \textit{tight} if $A = B$ in this definition, and \textit{unit norm} if $|| f_k ||_2 = 1, \forall k \in \mathcal{I}$. 

\end{definition} 

We define the \textit{coherence} $\mu$ of the frame to be the maximum correlation between any two distinct columns: $$\mu = \max_{i \ne j} \frac{|\langle f_i, f_j \rangle|}{||f_i||_2 \cdot ||f_j||_2}. $$

Designing frames with low coherence is a problem that has connections to a wide range of fields, including compressive sensing \cite{Donoho, Huo, Candes, Romberg, Tao, Tropp07}, spherical codes \cite{Delsarte, Heath}, LDPC codes \cite{Fan}, MIMO communications \cite{Pailraj, Bolcskei}, quantum measurements \cite{Eldar, Scott, RenesBlumeKohoutScottCaves}, etc. Frame theory has also made its mark as an interesting field in its own right, with a great collection of recent work by Casazza, Kutyniok, Fickus, Mixon, and many others \cite{CasazzaKutyniokFickusMixon, FickusMixonTremain, CasazzaArt, CasazzaKutyniok2, CasazzaKovacevic, CasazzaKutyniokLi}. 

Most often we will consider our frame vectors to be the columns $\{\vect{m}_i\}_{i = 1}^n$ of a matrix $\M = [\vect{m}_1, \vect{m}_2, \hdots, \vect{m}_n ]\in \C^{m \times n}$. We will speak of the \textit{coherence of $\M$} to be the coherence of the frame $\{\vect{m}_i\}$. The frame is tight if and only if $\M \M^* = \lambda \matr{I}_{m}$ where $\matr{I}_m$ is the $m \times m$ identity matrix and $\lambda$ is a scalar. Furthermore, $\lambda = A = B$ in (\ref{eqn:framedef}). If the $\{\vect{m}_i\}$ form a unit norm tight frame, then $\lambda = \frac{n}{m}$.

It is easy to see that any orthonormal basis for a vector space is a tight frame, and consequently a frame can be regarded as a generalization of an orthonormal set which may include more vectors than the dimension of the space. One important example of a tight frame that we will encounter is when the rows of $\M$ are a subset of the rows of the $n \times n$ discrete Fourier matrix: 
\begin{definition} 
Let $\omega = e^{\frac{2 \pi i}{n}}$, and let $\mathcal{F}$ be the discrete Fourier matrix, whose $(i, j)^{th}$ entry is $\omega^{ij}$. Let $\M = [\vect{m}_1, ..., \vect{m}_n] \in \C^{m \times n}$ such that the \textit{rows} of $\M$ are a subset of the rows of $\lambda \mathcal{F}$ for some scalar $\lambda \in \C$. Then $\{\vect{m}_i\}_{i = 1}^n$ is called a $\textit{harmonic frame}$. 
\end{definition} 
\textit{Remark:} The notion of a harmonic frame is actually more general than in this definition, as is explained in \cite{CasazzaKutyniok}, but the more general harmonic frames are also tight with equal norm elements. For our purposes, the above definition will suffice. We will touch on generalized harmonic frames in Section \ref{sec:abelianharmonicframes}, but there is a substantial collection of results on harmonic frames in \cite{ChienWaldron} and \cite{HayWaldron}.

Of great interest is when a tight frame is \textit{equiangular}: 
\vspace{.1in}
\begin{definition} 
A unit-norm frame $\{f_k\}_{k \in \mathcal{I}}$ is said to be equiangular if there is some constant $\alpha$ such that for any $i \ne j$, $|\langle f_i, f_j \rangle| = \alpha$. 

\end{definition}

The following theorem, known as the \textit{Welch bound} and based on the results of \cite{Welch}, provides a lower bound on the coherence of a frame: 

\begin{theorem} 
\label{thm:framebound} 
Let $\mathbb{E}$ be the field of real or complex numbers, and $\{f_k\}_{k = 1}^n$ be a unit-norm frame for $\mathbb{E}^m$. Then 
\begin{equation} 
\max_{i \ne j} | \langle f_i, f_j \rangle | \ge \sqrt{\frac{n - m}{m (n-1)}}, 
\end{equation} with equality if and only if $\{f_k\}_{k = 1}^n$ is tight and equiangular. 

\end{theorem} 

\begin{proof} 
This theorem is a classical result, and one of an even broader set of bounds \cite{Welch}. A quick proof which can be found in \cite{Heath} involves considering the eigenvalues of the Gram matrix $\matr{G}$ defined by $\matr{G}_{ij} = \langle f_i, f_j \rangle$. 
\end{proof}

\vspace{.1in} 
Thus, we would like to identify tight, equiangular frames for use in constructing matrices which achieve this lower bound. This problem arises in various contexts, for example line packing problems \cite{ConwayHarding}. It should be emphasized that such frames do not exist for all values of $m$ and $n$, so in general, we would also like to find ways to optimize the coherence by choosing $\M$ wisely from a cleverly designed class of matrices. Our approach will be to use the group frame construction proposed by Slepian \cite{Slepian} in the 1960s. Group frames have received a great deal of attention in recent years, notably in the substantial collection of work by Vale, Waldron, and others \cite{Vale, ValeWaldron2, ValeWaldron3, Waldron1, ChienWaldron, HayWaldron}. For an excellent review of the work in group frames, see \cite{CasazzaKutyniok}.

On one final note before proceeding, a common approach to produce a set of vectors with low correlation is to construct a set of Mutually Unbiased Bases (MUBs). Two bases $\{e_1, ..., e_m \}$ and $\{e'_1, ..., e'_m\}$ for $\C^d$ are mutually unbiased if each is orthonormal, and $| \langle e_i, e'_j \rangle | = \frac{1}{\sqrt{m}}$ for any $i$ and $j$. Algebraic constructions of up to $m+1$ MUBs are known in prime-power dimensions $m$, allowing for a number of vectors at most $m^2 + m$ \cite{Klappenecher, Bandyopadhyay, Wootters}. The frame constructions presented in this paper will at times outperform this coherence, though typically with a smaller number of vectors. More importantly, though, our frames do not require $m$ to be prime.

\section{Reducing  the Number of Distinct Inner Products in Tight Frames} 

In practice, constructing frames which are both tight and equiangular can prove difficult. It turns out, however, that we can expect reasonably low coherence from tight frames if we just require that the inner products between frame elements take on few distinct values, provided that each of these values arises the same number of times.

The following lemma, which is in some sense a generalization of the Welch bound, provides a bound on the coherence of a tight frame: 

\begin{lemma} 
Let $\{f_i\}_{i = 1}^{n} \subset \C^m$ be a unit-norm tight frame such that the absolute values of the inner products, $|\langle f_i, f_j \rangle|_{i \ne j}$, take on $r$ distinct values, each occurring the same number of times. Then the coherence $\mu$ of $\{f_i\}$ is at most a factor of $\sqrt{r}$ greater than the Welch bound. That is, 
\begin{align} 
\mu & \le \sqrt{r} \sqrt{\frac{n-m}{m(n-1)}}. 
\end{align} 
\label{lem:tightframecohbound} 
\end{lemma} 

\begin{proof} 
As a preliminary fact, Theorem 6.2 of \cite{BenedettoFickus} shows that the \textit{frame potential} $\sum_{i, j} | \langle f_i, f_j \rangle |^2$ is at least $\frac{n^2}{m}$ with equality if and only if the frame is tight. Let $\alpha_1, ..., \alpha_r$ be the distinct squared absolute values of the inner products, $\{|\langle f_i, f_j \rangle|^2\}_{i \ne j}$. Since each of the $\alpha_i$ occurs the same number of times as a squared inner product norm, we have that their arithmetic mean is equal to that of the $\{|\langle f_i, f_j \rangle|^2\}_{i \ne j}$, which is 
\begin{align} 
\frac{1}{r} \sum_{i = 1}^r \alpha_i & = \frac{1}{n(n-1)} \sum_{i \ne j} |\langle f_i, f_j \rangle |^2 = \frac{n-m}{m(n-1)},   
\end{align} 
 where the second equality follows from the preliminary frame potential result and the fact $\{f_i\}_{i = 1}^n$ is tight and unit-norm by assumption. 
 
Thus, since all the $\alpha_i$ are nonnegative we see that 
\begin{align} 
\mu^2 = \max_i \alpha_i \le \sum_{i = 1}^r \alpha_i = r \cdot \frac{n-m}{m(n-1)}, 
\end{align} 
from which the result follows. 
\end{proof} 

In light of Lemma \ref{lem:tightframecohbound}, our goal will be to construct a tight frame whose elements have very few inner product values between them, each of which occurs with the same multiplicity. In the following sections, we will present a group theoretic way to do this.

\section{Frames from Unitary Group Representations: Slepian Group Codes} 
\label{sec:unitmat} 

In \cite{Slepian}, Slepian proposed a method to construct low-coherence matrices by reasoning that the key to controlling the inner products between the columns was to reduce the number of distinct inner product values which arise. His construction, which has come to be known as a \textit{group frame}, has since been generalized (see, for example \cite{Vale} and \cite{CasazzaKutyniok}). On this note, let $\mathcal{U} = \{\U_1, \U_2, ..., \U_n\}$ be a (multiplicative) group of unitary matrices. We can equivalently view $\mathcal{U}$ as the image of a faithful, unitary representation of a group $\mathcal{G}$. In some works, e.g. \cite{Gabardo}, $\mathcal{U}$ is taken to be a group-like unitary operator system---the image of a projective representation---but normal representations will suffice for our purposes. Such representations exist for any finite group. 

Suppose that for each $i$, we have $\U_i \in \C^{m \times m}$ (or equivalently, $\mathcal{U}$ is the image of an $m$-dimensional representation). Let $\vv = [v_1, ..., v_m]^T \in \C^{m \times 1}$ be any vector, and let $\M$ be the matrix whose $i^{th}$ column is $\U_i \vv$: $\matr{M} = [\U_1 \vv, ..., \U_n \vv]$. The inner product between the $i^{th}$ and $j^{th}$ columns of $\M$ is $\langle \U_i \vv, \U_j \vv \rangle = \vv^* \U_i^* \U_j \vv$. Since $\mathcal{U}$ is a unitary group, we have $\U_i^* \U_j = \U_i^{-1} \U_j = \U_k$, for some $k \in \{1, ..., n\}$, so we can write $\langle \U_i \vv, \U_j \vv \rangle = \vv^* \U_k \vv$. In this manner, we have reduced the total number of pairwise inner products between the columns of $\M$ from ${n \choose 2}$ to $n-1$, the inner products parametrized by the non-identity elements of $\mathcal{U}$. Furthermore, we have the following: 

\begin{lemma} 
Let $\{\U_1, ..., \U_n\} \subset \C^{m \times m}$ be a set of distinct unitary matrices which form a group under multiplication, and let $\vv \in \C^{m \times 1}$ be a nonzero vector. Each of the values $\vv^* \U_k \vv$ occurs as the inner product between two columns of $\matr{M} = [\U_1 \vv, ..., \U_n \vv]$ the same number of times. 
\label{lem:samemultiplicityinnerprods}
\end{lemma} 
\begin{proof} 
For every choice of $\matr{U}_k$ and $\matr{U}_i$, there is a unique $\matr{U}_j$ such that $\matr{U}_i^{-1} \matr{U}_j = \U_k$. Thus, for each $\U_k$, there are $n$ pairs $(\U_i, \U_j)$ such that $\vv^* \U_i^* \U_j \vv = \vv^*\U_k \vv$. 
\end{proof}

\section{Abelian Groups and Harmonic Frames} 
\label{sec:abelianharmonicframes} 

For now, we will restrict ourselves to consider representations of abelian groups. Abelian groups are the simplest groups, in a sense, and have the special property that each of their irreducible representations is one-dimensional. Therefore, if $\mathcal{U}$ is the image of a representation of an abelian group, then all of the elements $\U_i$ can be simultaneously diagonalized by a change of basis matrix. Thus, we may assume without loss of generality that the $\U_i$ are \textit{diagonal} unitary matrices whose diagonal entries are powers of $\omega = e^{\frac{2 \pi i}{n}}$: 
\begin{align} 
\U_j = \diag(\omega^{k_{1, j}}, ..., \omega^{k_{m, j}}) \in \C^{m \times m}, \label{eqn:diagU} 
\end{align}   
where the $k_{i, j}$ are integers.

With each $\U_j$ in this form, the inner products between the normalized columns of $\M$ will take the form $$\frac{|\vv^* \U_j \vv|}{|| \vv ||_2^2} = \left| \sum_{i = 1}^m  \frac{\left|v_i\right|^2}{||\vv||_2^2} \omega^{k_{i, j}}\right|, $$ where $\vv = [v_1, ..., v_m]^T$. So we see that the entries of $\vv$ simply weight the diagonal entries of $\U_j$ in the above sum. In particular, without loss of generality, we may take the entries of $\vv$ to be real. Furthermore, it turns out that in order for our abelian group frame to be tight, all the entries $v_i$ must be of equal norm. This follows from Theorem 5.4 in \cite{CasazzaKutyniok}. On this note, we will consider the case where $\vv$ is the vector of all 1's, 
\begin{align} 
\vv = \mathbf{1}_m = [1, ..., 1]^T \in \C^{m \times 1}, \label{eqn:all1v}
\end{align} 
so that the above inner product norm becomes simply 
\begin{equation} \label{eqn:abelianinnerprod} \frac{|\vv^* \U_j \vv|}{|| \vv ||_2^2} = \frac{1}{m} \left|\sum_{i = 1}^m  \omega^{k_{i, j}}\right|. \end{equation}  

Notice that from Equation (\ref{eqn:abelianinnerprod}), we can see that the coherence of our final matrix would remain unchanged if we chose $\omega$ to be any other primitive $n^{th}$ root of unity. 

Let us examine the simple case where $\mathcal{U}$ is a cyclic unitary group, the most basic abelian group. That is, the elements of $\mathcal{U}$ can be written as the powers of a single unitary matrix $\U$ of order $n$: $$\mathcal{U} = \{\U, \U^2, ..., \U^{n-1}, \U^n = \matr{I_m}\}, $$ where $\matr{I_m}$ is the $m \times m$ identity matrix. We consider again choosing $\vv$ to be the vector of all ones as in (\ref{eqn:all1v}), and form the frame matrix $\M = [\vv, \U \vv, ..., \U^{n-1} \vv]$. Note that if we express our group elements diagonally as in (\ref{eqn:diagU}), 
\begin{align} 
\U = \diag(\omega^{k_{1}}, ..., \omega^{k_{m}}) \in \C^{m \times m}, \label{eqn:cyclicU} 
\end{align} 
we can see that the matrix $\M$ will take the form 
\begin{align} 
\M &= \begin{bmatrix} \vv & \U \vv & \hdots & \U^{n-1} \vv \end{bmatrix}  \\ 
& =  \begin{bmatrix} 1 &  \omega^{k_{1}} & \omega^{k_{1}\cdot 2} & \hdots & \omega^{k_{1}\cdot (n-1)}\\ 
						1 & \omega^{k_{2}} & \omega^{k_{2} \cdot 2} & \hdots & \omega^{k_{2}\cdot (n-1)}\\ 
						\vdots & \vdots & \vdots & \ddots & \vdots \\ 
						1& \omega^{k_{m}} & \omega^{k_{m} \cdot 2} & \hdots & \omega^{k_{m} \cdot (n-1)} \end{bmatrix}.   
\label{eqn:cyclicframemat} 
\end{align} 
If the $k_i$ are distinct this is a subset of rows of the discrete Fourier matrix, hence a harmonic frame.

For cyclic groups, the inner product between the columns $\U^{\ell_1} \vv$ and $\U^{\ell_2} \vv$, after normalizing the columns, will take the form $\frac{|\vv^* \U^{\ell_2 - \ell_1} \vv|}{||\vv||_2^2}$, which is the value of the inner product determined by $\U^{\ell_2 - \ell_1}$ in (\ref{eqn:abelianinnerprod}).

A general abelian group $G$ can be represented as follows: First express $G$ as a direct product of, say, $L$ cyclic groups of orders $n_1, ..., n_L$, so that $G \cong \frac{\Z}{n_1 \Z} \times ... \times \frac{\Z}{n_L \Z}$. Then let $\omega_1, ..., \omega_L$ be the corresponding primitive roots of unity: $\omega_j = e^{2 \pi i / n_j}$. Then we set $\U_j = \diag(\omega_j^{k_{1 j}}, ..., \omega_j^{k_{m j}})$, where we will assume that the $k_{ij}$ are distinct integers modulo $n_j$. The abelian group generated by the diagonal matrices $\{\U_1, ..., \U_L\}$ is isomorphic to $G$, and an arbitrary element will take the form $\U_1^{a_1} \U_2^{a_2} \hdots \U_L^{a_L}$, where $a_j \in \{0, ..., n_j - 1\}$. Our frame matrix $\M$ will then take the form $\M = [\hdots \left( \U_1^{a_1} \U_2^{a_2} \hdots \U_L^{a_L} \vv \right) \hdots]_{0 \le a_j \le n_j - 1}$.  

In this form, our previous cyclic frames clearly arise as subsets of the columns of $\M$. It turns out that these abelian frames are the generalized harmonic frames as described in \cite{CasazzaKutyniok}, up to a unitary rotation of $\vv$ or an equivalent group representation. The frame matrix $\M$ is a subset of rows of the Kronecker product  $\matr{A}_1 \otimes ... \otimes \matr{A}_L$, where $\matr{A}_j = \begin{bmatrix} \vv, \U_j \vv, ..., \U_j^{n_j - 1} \vv \end{bmatrix}$. As such, any of these frames will be tight.

\section{Equiangular Frames from Cyclic Group Representations} 

Let us examine the harmonic frame formed by the columns of $\M$ in (\ref{eqn:cyclicframemat}). \cite{Giannakis} classified the conditions on the $k_i$ under which this frame is equiangular. Since we know these frames are tight, this determines precisely when their coherence achieves the Welch lower bound of Theorem \ref{thm:framebound}. 

\begin{definition} 
Let $G$ be a group. A \textit{difference set} $K = \{k_1, ..., k_m\} \subset G$ is a set of elements such that every nonidentity element $g \in G$ occurs as a difference $k_i - k_j$ the same number of times. That is, the sets $A_g := \{(k_i, k_j) \in K \times K ~|~ k_i - k_j = g\}$ have the same size for $g \ne 0$. 
\end{definition}

\begin{theorem}[\cite{Giannakis} Equiangular Harmonic Frames] 
The harmonic frame formed by the columns of $\M$ in (\ref{eqn:cyclicframemat}) is equiangular if and only if the integers $k_i$ form a difference set in $\Z/n\Z$. 
\label{thm:diffsets} 
\end{theorem} 

\begin{proof} 
The proof follows from a simple but insightful Fourier connection. Let us define $A_t := \{(k_i, k_j) \in K \times K ~|~ k_i - k_j \equiv t \mod n \}$ for any $t \in \Z/n\Z$, and set $a_t := |A_t |$. Furthermore, if we index the columns as $\ell = 0, 1, ..., n-1$ then the inner product associated to the $\ell^{th}$ column takes the form $$c_{\ell} := \frac{\vv^* \U^{\ell} \vv}{||\vv||_2^2} =  \frac{1}{m} \sum_{k \in K} \omega^{\ell k}. $$ Since we are concerned only with the magnitude of $c_{\ell}$, we may consider the quantity $$\alpha_\ell := |c_\ell|^2 = \frac{1}{m^2} \left(\sum_{k \in K} \omega^{\ell k}\right)^* \left(\sum_{k \in K} \omega^{\ell k}\right)$$ $$ = \frac{1}{m^2} \sum_{k_i, k_j \in K} \omega^{\ell (k_i - k_j)}. $$ 

We can then write 
\begin{align} 
\alpha_\ell = \frac{1}{m^2} \sum_{t = 0}^{n-1} a_t \omega^{\ell t}, \label{eqn:alphaintermsofa} 
\end{align} 
which gives us a Fourier pairing between the $\alpha_\ell$ and the $a_t$ with inverse transform given by 
\begin{align} 
a_t = \frac{m^2}{n} \sum_{\ell = 0}^{n-1} \alpha_\ell \omega^{-t \ell}. \label{eqn:aintermsofalpha} 
\end{align}

$\M$ will be an equiangular tight frame precisely when all of the $\alpha_\ell$ are equal for $\ell \ne 0$, and from the Fourier pairing this will occur precisely when the $a_t$ are equal for $t \ne 0$, i.e., when the $k_i$ form a difference set. 
\end{proof}

This concept of tight equiangular frames arising from difference sets has since been generalized and elaborated \cite{DingFeng}, \cite{CasazzaKutyniok}, \cite{Waldron1}. \cite{Ding} showed how slightly relaxed forms of difference sets can produce frames which have coherence almost reaching the Welch Bound. Many of our results in the following sections can also be viewed as using more extensively-relaxed difference sets to produce low-coherence frames. Difference sets have been long studied and classified \cite{Baumert}, \cite{BethJungnickel}. They have found application in other fields as well, such as designing codes for DS-CDMA systems \cite{DingGolinKlove}, LDPC codes \cite{VasicMilenkovic}, sonar and synchronization \cite{Golomb}, and other forms of frame design \cite{Kaira}. 

While Theorem \ref{thm:diffsets} completely characterizes the optimal-coherence frames arising from representations of cyclic groups, it reveals that equiangular frames of the form (\ref{eqn:cyclicframemat}) are rather scarce, since the number of known difference sets is relatively small. In the following section, we will present a new strategy for selecting the integers $k_i$ which, while not always producing an equiangular frame, does yield frames with few distinct inner product values and provable low coherence.

\section{Cyclic Groups of Prime Order } 
\label{sec:cyclicprimegrp} 

We have already managed to cut down the number of distinct inner products between columns from ${n \choose 2}$ to $n-1$, simply by using a unitary group to generate our columns. For cyclic groups, however, we can reduce this number even more. We first consider the case where $n$ is prime. Let $G = \left( \Z/n\Z \right)^\times$, the multiplicative group of the integers modulo $n$. As usual, we identify the elements of $\Z/n\Z$ with the integers $0, 1, ..., n-1$. Since $n$ is assumed to be a prime, $G$ is itself a cyclic group, and consists of the $n-1$ nonzero elements of $G$. Now let us choose $m$ to be any divisor of $n-1$, and set $r := \frac{n-1}{m}$. Since $G$ is cyclic, it has a unique subgroup $K$ of order $m$ consisting of the distinct $r^{th}$ powers of the elements of $G$. In fact, if $g$ is any generator for $G$, then $K$ will be generated by $k := g^\frac{n-1}{m}$. Now, if we write out the elements of $K$ as $\{k_1, ..., k_m\}$ (or equivalently in terms of a single generator $k$ as $\{1, k, k^2, ..., k^{m-1}\}$), we can form our generator matrix $\U$ as in (\ref{eqn:cyclicU}), choosing $\omega^{k_i}$ to be the $i^{th}$ diagonal term. Note that since $K$ consists of elements relatively prime to $n$, then for each $i$, $\omega^{k_i}$ has multiplicative order $n$. It follows that $\U$ also has order $n$ and generates the cyclic group $\mathcal{U} = \{\U^\ell\}_{\ell = 0}^{n-1} \cong \Z/n\Z$. 

It turns out that this construction not only reduces the number of distinct inner product values between our columns, but it maintains the property that each such value occurs with the same multiplicity: 

\begin{theorem} Let $n$ be a prime and $m$ any divisor of $n-1$. Take $K = \{k_1, ..., k_m\}$ to be the unique (cyclic) subgroup of $G = (\Z/n\Z)^\times$ of size $m$. Set $\omega = e^{\frac{2 \pi i}{n}}$, $\vect{v} = \frac{1}{\sqrt{m}}[1, ..., 1]^T \in \R^{m}$, and $\matr{U} = \diag(\omega^{k_1}, ..., \omega^{k_m})$. Then the columns of $\matr{M} =  [\vv, \matr{U} \vv, ..., \matr{U}^{n-1} \vv]$ form a unit-norm tight frame with at most $\frac{n-1}{m}$ distinct inner product values between its columns, each occurring with the same multiplicity. 
\label{thm:nprimealgorithm} 
\end{theorem} 

 \begin{proof} 
For any integer $\ell$ in the set $\{1, ..., n-1\}$, the inner product corresponding to $\U^\ell$ (as in Equation (\ref{eqn:abelianinnerprod})) will take the following form: 

\begin{equation} \label{eqn:cyclicinnerprod} \frac{|\vv^* \U^\ell \vv|}{|| \vv ||_2^2} = \frac{1}{m} \left|\sum_{i = 1}^m  \omega^{\ell \cdot k_{i}}\right|. \end{equation} 

Notice the exponents of $\omega$ appearing in the above summation can be taken modulo $n$, since $\omega$ is an $n^{th}$ root of unity, and are then simply the elements of the $\ell^{th}$ coset of $K$ in $G$, $\ell K = \{\ell \cdot k_1, ..., \ell \cdot k_m\}$. The set of all cosets of $K$ in $G$ is denoted $G/K$. From elementary group theory, we know that the distinct cosets of $K$ form a disjoint partition of $G$, so the number of distinct cosets of $K$ in $G$ is the quotient of their sizes: $|G/K| = \frac{|G|}{|K|} = \frac{n-1}{m}$. Thus, the total number of distinct pairwise inner products that we now must control is $\frac{n-1}{m}$. 

It only remains to show that each of the $\frac{n-1}{m}$ inner products occurs the same number of times. Let $\{\ell_1, ..., \ell_r\}$ be a complete set of coset representatives for $K$ in $(\Z/n\Z)^\times$. Here, $r$ is simply $\frac{n-1}{m}$. Then every element in $\{1, ..., n-1\}$ can be written uniquely as a product $\ell_i k_j$, and from Lemma \ref{lem:samemultiplicityinnerprods} the $n-1$ inner products $\vv^* \U^{\ell_i k_j} \vv$ all arise the same number of times. As described above, the $\frac{n-1}{m}$ distinct inner product values correspond to the cosets of $K$, i.e., for a fixed $\ell_i$ the $m$ inner products $\vv^* \U^{\ell_i k_1}\vv, ..., \vv^*\U^{\ell_i k_m}\vv$ will give rise to one of the distinct inner product values. Thus, since each distinct value corresponds to $m$ inner products, each arising the same number of times, our result is proved. 
\end{proof}

We emphasize the power of this construction in reducing the number of inner products that we must control in order to maintain low matrix coherence. Since we are free to choose $m$ to be \textit{any} divisor of $n-1$, then for properly chosen matrix dimensions, we can reasonably create matrices with just two or three distinct values of inner products between columns. In practice, this often creates matrices with remarkably low coherence, far outmatching that of any known randomly-generated matrices. In Table \ref{table:coherencelist}, we compare the coherences of the ``Group Matrices" from our construction with those of randomly-generated complex Gaussian matrices and matrices designed by randomly selecting $m$ rows from the $n \times n$ Fourier matrix. (This latter construction is equivalent to randomly selecting the exponents $k_i$ in our cyclic generator matrix $\U$ in (\ref{eqn:cyclicU}).) For convenience, we also list the lower bound on coherence from Theorem \ref{thm:framebound}, and we underline the coherences which achieve this bound. Figure \ref{fig:innerprod_499_166} illustrates explicitly the inner products for a random Fourier matrix vs. a Group matrix. 

\begin{table}[ht]
\caption{Coherences for Random and Group Matrices (for $n$ a Prime)} 
\centering 
\begin{tabular}{c p{1.5cm}  p{1cm}   p{1cm} c }
\hline \hline 
$(n, m)$ & Complex Gaussian & Random Fourier & Group Matrix & $\sqrt{\frac{n - m}{m (n-1)}}$ \\ [0.5ex] 
\hline 
(251, 125) & .2677 & .1996 & \underline{.0635} & \underline{.0635} \\ 
(499, 166) & .3559 & .1786 & .0888 & .0635 \\ 
(499, 249) & .2226 & .1736 & \underline{.0449} & \underline{.0449} \\ 
(503, 251) & .2137 & .1533 & \underline{.0447} & \underline{.0447} \\ 
(521, 260) & .2208 & .1504 & .0458 & .0439 \\
(521, 130) & .3065 & .2376 & .1175 & .0761 \\ 
(643, 321) & .2034 & .1627 & \underline{.0395} & \underline{.0395} \\ 
(643, 214) & .2274 & .1978 & .0755 &.0559 \\ 
(701, 175) & .2653 & .2316 & .0687 & .0655 \\ 
(701, 350) & .1788 & .1326 & .0393 & .0379 \\ 
(1009, 504) & .1565 & .1147 & .0325 & .0315 \\ 
(1009, 336) & .2086 & .1384 & .0597 & .0446 \\ 
(1009, 252) & .2287 & .1631 & .0846 & .0546 \\ 
\hline 
\end{tabular} 
\label{table:coherencelist} 
\end{table}

\begin{figure}
\begin{center}
 \includegraphics[height = 230pt, width=230pt]{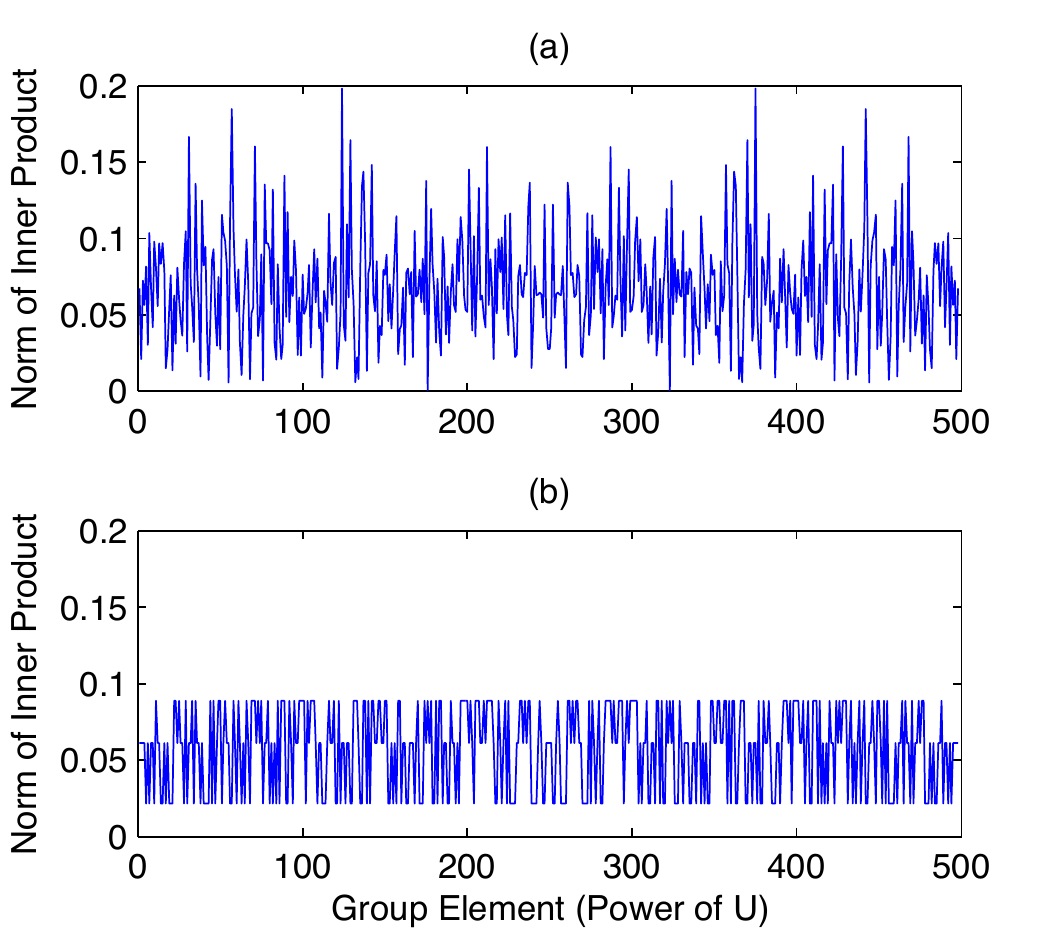}
 \caption{The norms of the inner products associated to each group element for (a) randomly-chosen $K$, and (b) $K$ selected to be a subgroup of $(\Z/p\Z)^\times$ of index 3. Here, $n = p = 499, m = 166$. In (b), as expected, there are only three distinct values of the inner products between distinct, normalized columns. }
 \label{fig:innerprod_499_166}
\end{center}
\end{figure}

\section{Sharper Bounds on Coherence for Frames from Cyclic Groups of Prime Order} 

In the special case where we construct our frame as in Theorem \ref{thm:nprimealgorithm} (using Slepian's approach with a group $\mathcal{U} \cong \Z/n\Z$ and $n$ prime), we have a great deal of underlying algebraic structure in our frame. So it should come as no surprise that we can derive sharper bounds on our coherence and even compute it exactly in some cases.

As before, let $m$ be a divisor of $n-1$, and take $K = \{k_1, ..., k_m\}$ to be the unique subgroup of $(\Z/n\Z)^\times$ of size $m$. Define $r := \frac{n-1}{m}$, which is the number of distinct inner product values. If $r$ is small, it becomes relatively simple to analyze these values. For example:

\begin{theorem}[r = 2] 
Let $n$ be a prime, $m$ a divisor of $n-1$, and $\omega = e^{\frac{2 \pi i}{n}}$. Let $K = \{k_1, ..., k_m\}$ be the unique subgroup of $(\Z/n\Z)^\times$ of size $m$, and set $\U = \diag(\omega^{k_1}, ..., \omega^{k_m}) \in \C^{m \times m}$, $\vv = \frac{1}{\sqrt{m}} [1, ..., 1]^T \in \C^{m \times 1}$, and $\matr{M} = [\vv, \U \vv, ..., \U^{n-1} \vv]$. 

If $r := \frac{n-1}{m} = 2$, there are two distinct inner product values between the columns of $\matr{M}$, both of which are real. If $n-1$ is divisible by 4, these inner products are $\frac{-1 \pm \sqrt{1 + 2m}}{2m}$. In this case, $\matr{M}$ has coherence $\sqrt{\frac{n - m - \frac{1}{2}}{m(n-1)}} + \frac{1}{2m}$. 

If $n-1$ is not divisible by 4, then the columns of $\matr{M}$ form an equiangular frame. The two inner products are $\pm \sqrt{\frac{1}{m} \left(\frac{1}{2} + \frac{1}{2m}\right)} $, and the coherence is $\sqrt{\frac{n-m}{m(n-1)}}$. 
\label{thm:r2innerprods} 
\end{theorem} 

\begin{proof}
We will hold off on the details of the proof until Appendix \ref{sec:r2} aside from mentioning that it is related to the connection made by Xia et al \cite{Giannakis} between tight equiangular harmonic frames and difference sets. In fact, in the case where $n-1$ is not divisible by 4, $K$ forms a known difference set in $\Z/n\Z$. If we view $\Z/n\Z$ as the additive group of $\mathbb{F}_n$, this particular case also overlaps with the tight equiangular frames classified in Theorem 3 of \cite{DingFeng}. 
\end{proof}

As the number $r$ of inner products increases, it becomes more complicated to explicitly compute their values or even just the coherence of the resulting frame. While there were only two cases to consider when $r = 2$, there are many more even for $r$ as low as 3. We can, however, exploit the algebraic structure of our frames to yield bounds on their coherence which in practice prove to be nearly tight. 

\begin{theorem}[r = 3] 
Let $n$ be a prime, $m$ a divisor of $n-1$, and $\omega = e^{\frac{2 \pi i}{n}}$. Let $K = \{k_1, ..., k_m\}$ be the unique subgroup of $(\Z/n\Z)^\times$ of size $m$, and set $\U = \diag(\omega^{k_1}, ..., \omega^{k_m}) \in \C^{m \times m}$, $\vv = \frac{1}{\sqrt{m}} [1, ..., 1]^T \in \C^{m \times 1}$, and $\matr{M} = [\vv, \U \vv, ..., \U^{n-1} \vv]$. 

If $r := \frac{n-1}{m} = 3$, then the coherence of $\matr{M}$ will satisfy 
\begin{align} 
\mu \le \frac{1}{3} \left(2\sqrt{\frac{1}{m} \left(3 + \frac{1}{m} \right)} + \frac{1}{m}\right) \approx \sqrt{\frac{4}{3m}}, 
\end{align} 
and for large enough $m$, we will \textit{asymptotically} have the following lower bound on coherence: 
\begin{align} 
\mu \ge \frac{1}{\sqrt{m}} \text{~(asymptotically), }  
\end{align} 
which is strictly greater than the Welch bound. 
\label{thm:r3cohbounds} 
\end{theorem} 

\begin{proof} 
We present the proof in Appendix \ref{sec:r3}. 
\end{proof}

From Theorem \ref{thm:r3cohbounds} we see that unlike when $r=2$, we can never hope to achieve the Welch bound with these frames when $r=3$. But this is not a trend, for our frames will again be able to achieve the Welch bound for certain higher values of $r$, including $r = 4$ and $r = 8$. This again relates to the connection with difference sets from \cite{Giannakis}. As a result, the lower bound on coherence in Theorem \ref{thm:r3cohbounds} does not generalize to all values of $r$. Fortunately, the upper bound does:

\begin{theorem}[General $r$] 
Let $n$ be a prime, $m$ a divisor of $n-1$, and $\omega = e^{\frac{2 \pi i}{n}}$. Let $K = \{k_1, ..., k_m\}$ be the unique subgroup of $(\Z/n\Z)^\times$ of size $m$, and set $\U = \diag(\omega^{k_1}, ..., \omega^{k_m}) \in \C^{m \times m}$, $\vv = \frac{1}{\sqrt{m}} [1, ..., 1]^T \in \C^{m \times 1}$, and $\matr{M} = [\vv, \U \vv, ..., \U^{n-1} \vv]$. 

If $r := \frac{n-1}{m}$, then the coherence $\mu$ of $\matr{M}$ satisfies the following upper bound: 
\begin{align} 
\mu & \le \frac{1}{r} \left( (r-1) \sqrt{\frac{1}{m} \left(r + \frac{1}{m}\right)} + \frac{1}{m} \right). 
\end{align} 
\label{thm:coherenceupperbound} 
\end{theorem} 

\begin{proof} 
This theorem will be proved in Appendix \ref{sec:generalr}. 
\end{proof} 

This bound is strictly lower than the one from Lemma~\ref{lem:tightframecohbound}, which applies to all tight frames. In fact, when $n>2$, we can find an even lower bound on the coherence of our frames constructed in Theorem \ref{thm:nprimealgorithm} which surprisingly depends only on whether $m$ is odd: 

\begin{theorem}[$m$ odd] 
Let $n$ be an odd prime, $m$ a divisor of $n-1$, and $\omega = e^{\frac{2 \pi i}{n}}$. Let $K = \{k_1, ..., k_m\}$ be the unique subgroup of $(\Z/n\Z)^\times$ of size $m$, and set $\U = \diag(\omega^{k_1}, ..., \omega^{k_m}) \in \C^{m \times m}$, $\vv = \frac{1}{\sqrt{m}} [1, ..., 1]^T \in \C^{m \times 1}$, and $\matr{M} = [\vv, \U \vv, ..., \U^{n-1} \vv]$. Set $r := \frac{n-1}{m}$. 

If $m$ is odd, then the coherence of $\matr{M}$ is upper-bounded by 
\begin{align} 
\mu & \le \frac{1}{r} \sqrt{\left(\frac{1}{m} + \left(\frac{r}{2} - 1 \right) \beta \right)^2 + \left(\frac{r}{2}\right)^2 \beta^2},  
\end{align} 
where $\beta = \sqrt{\frac{1}{m} \left( r + \frac{1}{m} \right) }$. 
\label{thm:coherenceupperboundmodd} 
\end{theorem} 

\begin{proof} 
We delay the proof of this theorem until Appendix \ref{sec:generalr}. 
\end{proof}

It is worth noting that this latter bound has no analog in the $r = 3$ situation because $m$ must always be even in that case. We explain the reason for this in the sequel, and give an alternate classification for exactly when this latter coherence bound applies. We illustrate the upper and lower bounds for $r=3$ in Figure \ref{fig:r3cohbounds} and the two upper upper bounds from Theorem \ref{thm:coherenceupperboundmodd} for when $r = 4$. When $r = 4$, we can also derive different lower bounds on the coherence for when $m$ is even or odd, and together with the two upper bounds from the theorems they form two non-overlapping regions in which the coherences can fall in the graph. While these regions will exist for every $r$, they will sometimes overlap (that is, the lower bound on coherence for $m$ even could be less than the upper bound for $m$ odd).

\begin{figure}
\begin{center}
 \includegraphics[height = 220pt, width=260pt]{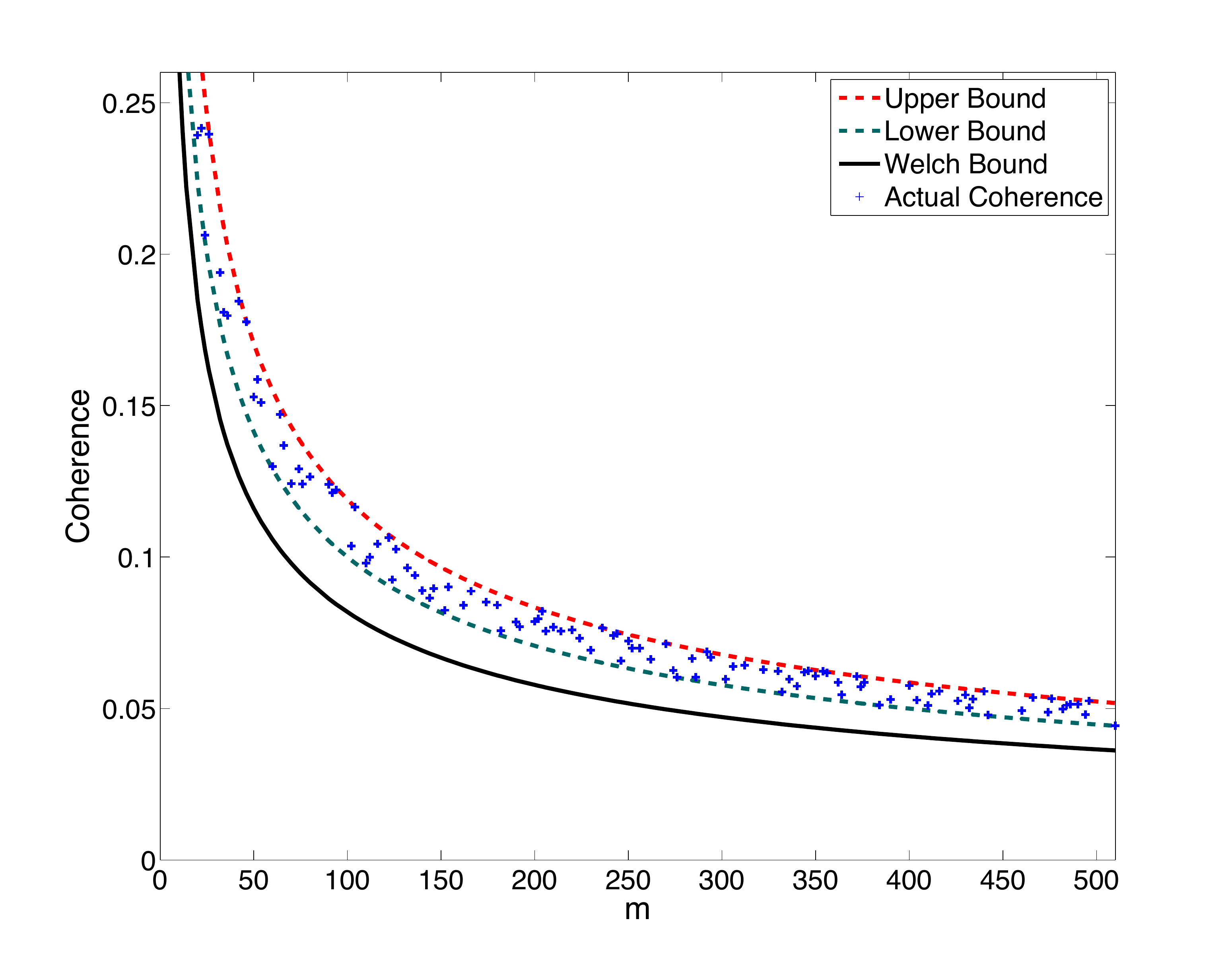}
 \caption{The upper and lower bounds on coherence for $r = 3$. }
 \label{fig:r3cohbounds}
\end{center}
\end{figure}

\begin{figure}
\begin{center}
\includegraphics[height = 220pt, width=260pt]{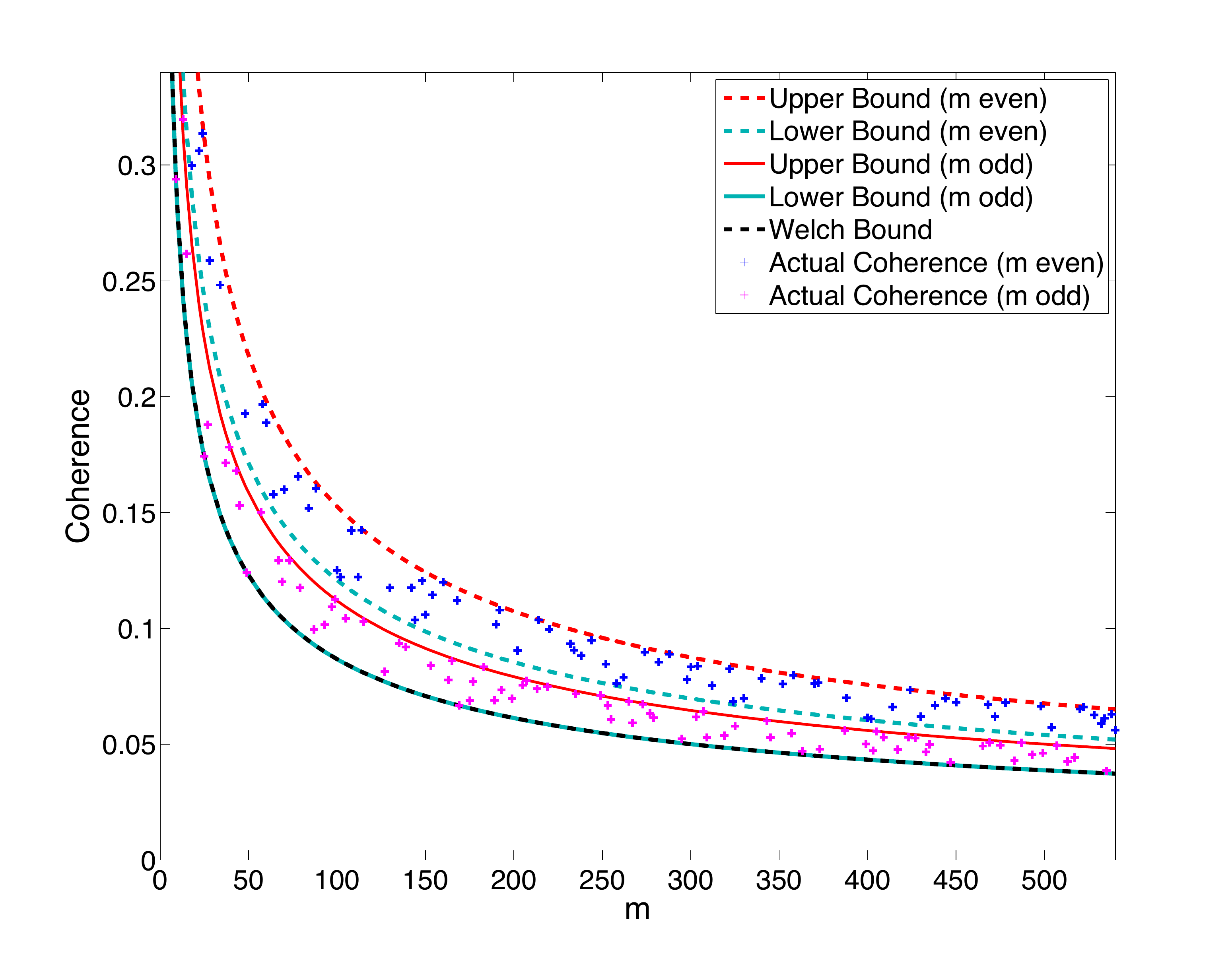}
\caption{The upper and lower bounds on coherence for $r = 4$. }
 \label{fig:r4cohbounds}
\end{center}
\end{figure}

\section{Generalized Dihedral Groups} 
Rather than dwell on clever constructions of general abelian groups, let us instead investigate what changes when $\mathcal{U}$ is nonabelian. In this case the irreducible representations at our disposal will no longer all be one-dimensional, so we will no longer have all the matrices $\U_i$ be simultaneously diagonal. Consequently, it is no longer clear that we can restrict our vector $\vv$ to be real-valued. 

One simple class of nonabelian groups is that of semidirect products of cyclic groups. On this note, consider the following group presentation (which arises in \cite{Hassibi}): 
\begin{align} 
G_{n, r} = \langle \sigma, \tau~|~ \sigma^n = 1, \tau^D = 1, \tau \sigma \tau^{-1} = \sigma^r \rangle. \label{eqn:Gnr} 
\end{align} 
Here, $n$ and $r$ are relatively prime integers, and $D$ is the multiplicative order of $r$ modulo $n$. $G_{n, r}$ is precisely a semidirect product in the form $\frac{\Z}{n\Z} \rtimes \frac{\Z}{D \Z}$, and if we take $D = 2$ and $r = n-1$, we see that we obtain the familiar dihedral group $D_{2n}$. 

There are $n \cdot D$ group elements in $G_{n, r}$, each of which can be written in the form $\sigma^a \tau^b$ for some integers $0 \le a < n$ and $0 \le b < D$. $G_{n, r}$ has an irreducible representation in the form 
\begin{align} 
\sigma \mapsto \matr{S} & := \begin{bmatrix} \omega & &  &\\ & \omega^{r} & & \\ & & \ddots & \\ & & & \omega^{r^{D - 1}} \end{bmatrix} \in \C^{D \times D},\\  
\tau \mapsto \matr{T} & := \begin{bmatrix} & 1 &  & & \\ & & 1 & & \\ & & &\ddots &  \\ & & & & 1 \\ 1 & & & &  \end{bmatrix} \in \C^{D \times D}, \label{eqn:sigmataueqn1} 
\end{align} 
where $\omega = e^\frac{2 \pi i}{n}$ (see again \cite{Hassibi}). The informed reader might note that this representation is quite similar to that of Heisenberg groups, which have been extensively applied to the construction of frames \cite{BosWaldron, Khatirinejad, RenesBlumeKohoutScottCaves, ScottGrassi}. Our following methods can be conceivably adjusted for use with Heisenberg frames as well. 

In order to construct our frames, we would like to follow the example of our previous construction in Theorem \ref{thm:nprimealgorithm} by selecting a representation for $G_{n, r}$ of the form 
\begin{align} 
\sigma \mapsto [\sigma] := \begin{bmatrix} \matr{S}^{k_1} & & \\ & \ddots & \\ & & \matr{S}^{k_m} \end{bmatrix}, ~ \tau \mapsto [\tau] := \begin{bmatrix} \matr{T} & &  \\ & \ddots & \\  & & \matr{T} \end{bmatrix}, \label{eqn:sigmataueqn2} 
\end{align} 
where $m$ and the $k_i$ are cleverly chosen integers. Then we will select a vector $\vv \in \C^{Dm \times 1}$ and take our frame to be the columns of the matrix $\matr{M} := [ \hdots [\sigma]^a [\tau]^b \vv \hdots]_{0 \le a < n, ~ 0 \le b < D}$. We must require that the greatest common divisor between the $k_i$ is relatively prime to $n$ in order for the columns to be distinct, and again we satisfy this by taking $n$ to be prime. Note that in our above notation, this will be a $Dm$-dimensional representation, so our resulting frame matrices will have dimensions $Dm \times Dn$. 

At this point, we can see that in order to minimize coherence we must deviate from our original construction, for if we were to set $\vv$ to the vector $\matr{1}$ of all ones it would be fixed by $[\tau]^b$ for any $b$, and the inner product corresponding to $[\tau]^b$ would be 1. We must therefore be more clever in how we construct $\vv$. A natural form for $\vv$ would be to find some $D$-dimensional vector $\w = [w_1, ..., w_D]^T \in \C^{D \times 1}$ and set $\vv$ equal to the periodic vector $\vv = \begin{bmatrix} \w^T & \w^T & ...& \w^T \end{bmatrix}^T \in \C^{Dm \times 1}$. The question now becomes how to choose $\w$? 

In order to preserve as much of the structure from our previous construction as possible, we would like each entry of $\w$ to have the same norm. This will ensure that the inner products corresponding to the elements $[\sigma]^a$ will have the same values as those in our previous construction from Theorem \ref{thm:nprimealgorithm} corresponding to when $\mathcal{U}$ was the cyclic group $\Z/n\Z$ generated by $[\sigma]$. Let us require that $w_d$ be unit norm for each $d$, and consider attempting to force $\w$ to satisfy the constraint that 
\begin{align} 
\w^* \matr{T}^b \w = \sum_d w^*_d w_{d+b} = 0, ~ \forall b 
\end{align} 
where the indices are taken modulo $D$. It turns out that we can satisfy all our restrictions on $\w$ by selecting its indices to form a \textit{Zadoff-Chu} (ZC) sequence \cite{Chu, Kahn}: 

\begin{align} 
w_d = 
\begin{cases} 
e^\frac{i \pi d^2}{D}, & \text{ if $D$ is even} \\ 
e^\frac{i \pi d(d+1)}{D}, & \text{ if $D$ is odd} 
\end{cases} 
\label{eqn:ZCsequence}
\end{align} 

This is a well-known constant amplitude zero autocorrelation (CAZAC) sequence. Our frame elements will now take the form 
\begin{align} 
[\sigma]^a [\tau]^b \vv = \begin{bmatrix} \matr{S}^{a k_1} \w_{d+b} \\ \vdots \\ \matr{S}^{a k_m} \w_{d+b} \end{bmatrix}, 
\end{align}
where $\w_{d + b} = \matr{T}^b \w$ denotes the vector obtained by cyclically shifting the entries of $\w$ by $b$ positions. Thus, as the notation would suggest, the $d^{th}$ entry of $\w_{d + b}$ is $w_{d+b}$. (Note that by this notation, $\w_d$ is simply $\w$ itself). Our inner products will take the form 
\begin{align} 
\frac{\vv^* [\sigma]^a [\tau]^b \vv}{||\vv||_2^2} = \frac{1}{m \cdot D} \sum_{j = 1}^m \w_d^* \matr{ S}^{a k_j} \w_{d + b}. \label{eqn:dihedralinnerprod} 
\end{align} 

Our new frames remain tight:

\begin{theorem} 
Let $n$ and $r$ be relatively prime integers, and $D$ the order of $r$ modulo $n$. Let $[\sigma] \in \C^{Dm \times Dm}$ and $[\tau] \in \C^{Dm \times Dm}$ be the generating matrices for $G_{n, r}$ defined in (\ref{eqn:sigmataueqn1}) and (\ref{eqn:sigmataueqn2}). If $\w = [w_1, ..., w_D]^T \in \C^{D \times 1}$ is a ZC-sequence (\ref{eqn:ZCsequence}), and $\vv = \begin{bmatrix} \w^T & ... & \w^T \end{bmatrix}^T \in \C^{Dm \times 1}$, then the the columns of $\M = \begin{bmatrix} \hdots & [\sigma]^a [\tau]^b \vv & \hdots \end{bmatrix} \in \C^{Dm \times Dn}$ form a tight frame. 
\end{theorem} 

\begin{proof} 
This result follows from a direct calculation, but can also be deduced from Theorem 5.4 of \cite{CasazzaKutyniok} since all the representations are of the same dimension and the corresponding components $\w_d$ of $\vv$ all have the same norm.  
\end{proof} 

Exploiting the properties of our construction, we can bound the coherence of our new frames by that of our original frames arising from representations of cyclic groups. 

\begin{theorem} 
Let $n$ be an integer, and $k_1, ..., k_m$ distinct integers modulo $n$ whose greatest common divisor is relatively prime to $n$. Take $r$ an integer relatively prime to $n$, and $D$ the multiplicative order of $r$ modulo $n$. Set $\omega = e^{\frac{2 \pi i}{n}}$. Consider the two frames: 
\begin{enumerate} 
\item The columns of the ``cyclic frame'' $\matr{M}_1 = [\vv_1, \U \vv_1, \hdots, \U^{n-1} \vv_1] \in \C^{m \times n}$ where $\U = \diag(\omega^{k_1}, ..., \omega^{k_m}) \in \C^{m \times m}$ and $\vv_1 = [1, ..., 1]^T \in \C^{m \times 1}$.  

\item The columns of the ``generalized dihedral frame'' $\matr{M}_2 = \begin{bmatrix} \hdots & [\sigma]^a [\tau]^b \vv_2 & \hdots \end{bmatrix} \in \C^{Dm \times Dn}$ where $\vv_2 = \begin{bmatrix} \w^T & ... & \w^T \end{bmatrix}^T \in \C^{Dm \times 1}$ and $\w = [w_1, ..., w_D]^T \in \C^{D \times 1}$ is a ZC-sequence. 
\end{enumerate} 
If $\mu^{cyc}_K$ is the coherence of the cyclic frame $\matr{M}_1$ and $\mu^{D}_K$ the coherence of the generalized dihedral frame $\matr{M}_2$, then $\mu^{D}_K \le \mu^{cyc}_K$. 
\label{thm:gendihedralcyclicbound} 
\end{theorem} 

\begin{proof} 
From (\ref{eqn:dihedralinnerprod}), we see that the inner products for the generalized dihedral representation will take the form 
\begin{align} 
\frac{\vv^* [\sigma]^a [\tau]^b \vv}{||\vv||_2^2}  & = \frac{1}{m \cdot D} \sum_{k \in K} \sum_{d} w_d^* w_{d + b}\omega^{k a r^{d-1}} \\ 
& = \frac{1}{m \cdot D}\sum_{d} w_d^* w_{d + b} \sum_{k \in K} \omega^{k a r^{d-1}} \\ 
& = \frac{1}{m \cdot D}\sum_{d} w_d^* w_{d + b} \sum_{k \in K} \omega^{k a'}, \label{eqn:gendihedralcosetinnerprods} 
\end{align} 
where $a' = a r^{d-1}$. Furthermore, 
\begin{align} 
\frac{|\vv^* [\sigma]^a [\tau]^b \vv|}{||\vv||_2^2} & \le  \frac{1}{m \cdot D}\sum_{d} \left| w_d^* w_{d + b} \sum_{k \in K} \omega^{k a'} \right| \\ 
& = \frac{1}{m \cdot D}\sum_{d} \left|\sum_{k \in K} \omega^{k a'} \right| \\ 
& \le \frac{1}{m \cdot D} \sum_{d} m \mu^{cyc}_K = \mu^{cyc}_K, 
\end{align} 
so $\mu^{D}_K \le \mu^{cyc}_K$. 
\end{proof}

Theorem \ref{thm:gendihedralcyclicbound} allows us to bound the coherence of our generalized dihedral frames using the same bounds from Theorems \ref{thm:coherenceupperbound} and \ref{thm:coherenceupperboundmodd}: 

\begin{corollary} 
Let $n$ be a prime and $m$ a divisor of $n-1$, and let $K = \{k_1, ..., k_m\}$ be the unique subgroup of $(\Z/n\Z)^\times$ of size $m$. Set $s = \frac{n-1}{m}$. Take $r$ an integer relatively prime to $n$, and $D$ the multiplicative order of $r$ modulo $n$. 

Let $[\sigma] \in \C^{Dm \times Dm}$ and $[\tau] \in \C^{Dm \times Dm}$ be the generating matrices for $G_{n, r}$ defined in (\ref{eqn:sigmataueqn1}) and (\ref{eqn:sigmataueqn2}). If $\w = [w_1, ..., w_D]^T \in \C^{D \times 1}$ is a ZC-sequence (\ref{eqn:ZCsequence}), and $\vv = \begin{bmatrix} \w^T & ... & \w^T \end{bmatrix}^T \in \C^{Dm \times 1}$, then the the columns of $\M = \begin{bmatrix} \hdots & [\sigma]^a [\tau]^b \vv & \hdots \end{bmatrix} \in \C^{Dm \times Dn}$ have at most $D \cdot \frac{n-1}{m}$ distinct inner product values between them, and the coherence $\mu$ of $\matr{M}$ is bounded by 
\begin{align} 
\mu & \le \frac{1}{s} \left( (s-1) \sqrt{\frac{1}{m} \left(s + \frac{1}{m}\right)} + \frac{1}{m} \right). \label{eqn:gendihedralcycbound1} 
\end{align} 

If $m$ is odd, then the coherence of $\matr{M}$ is upper-bounded by 
\begin{align} 
\mu & \le \frac{1}{s} \sqrt{\left(\frac{1}{m} + \left(\frac{s}{2} - 1 \right) \beta \right)^2 + \left(\frac{s}{2}\right)^2 \beta^2},  \label{eqn:gendihedralcycbound2} 
\end{align} 
where $\beta = \sqrt{\frac{1}{m} \left( s + \frac{1}{m} \right) }$. 

\label{cor:gendihedralcoherenceupperbounds} 
\end{corollary} 

\begin{proof} 
From (\ref{eqn:gendihedralcosetinnerprods}), we can write out the inner product corresponding to the group element  $\sigma^a \tau^b$ in the form 
\begin{align} 
\frac{\vv^* [\sigma]^a [\tau]^b \vv}{||\vv||_2^2} = \frac{1}{m \cdot D}\sum_{d} w_d^* w_{d + b} \sum_{k \in K} \omega^{k a'},  
\end{align} 
where $a' = a r^{d-1}$. In this form, we see that for each value of $d$ in the summation, there are $\frac{n-1}{m}$ possible distinct inner product values associated to the different cosets $a' K$, so there are at most $D\frac{n-1}{m}$ possible values. The last two bounds (\ref{eqn:gendihedralcycbound1}) and (\ref{eqn:gendihedralcycbound2}) follow from Theorem \ref{thm:gendihedralcyclicbound} and the bounds given from Theorems \ref{thm:coherenceupperbound} and \ref{thm:coherenceupperboundmodd}. 
\end{proof}

In the case of regular dihedral groups ($D = 2$), our $\w$ becomes $[1, i]^T$, and we can readily calculate our inner products to be 
\begin{footnotesize}
\begin{align} 
\nonumber  \frac{\vv^* [\sigma]^\ell \vv}{||\vv||_2^2}    = \operatorname{Re} \left( \frac{1}{m} \sum_{j = 1}^m \omega^{\ell k_j}\right)~,~ \frac{\vv^* [\sigma]^\ell [\tau] \vv}{||\vv||_2^2}  =  \operatorname{Im} \left( -\frac{1}{m} \sum_{j = 1}^m  \omega^{\ell k_j} \right). 
\end{align} 
\end{footnotesize}

As we can clearly see, each of these has magnitude bounded by that of the corresponding inner product in the cyclic counterpart, $\left|\frac{1}{m} \sum_{j = 1}^m  \omega^{\ell k_j}\right|$. In general, the dihedral coherence could be substantially smaller than the corresponding cyclic coherence. Most importantly, by extending to generalized dihedral groups, we allow for frame matrices $\M$ with a greater variety of dimensions. In particular, the number of columns ($nD$) no longer need be prime.

\section{Conclusion} 
We have presented a method to select a set of representations of a finite cyclic group to construct tight, unit-norm group frames  such that the frame elements take on very few distinct pairwise inner product values. Our construction ensures that each such inner product value arises the same number of times, allowing us to derive upper bounds on the coherence of the frames which approach the Welch lower bound. In certain cases, our construction has yielded instances of previously known tight, equiangular frames which achieve the Welch bound. We have then demonstrated how our method can be applied to constructing tight group frames from abelian and generalized dihedral groups to obtain a richer set of frames of different sizes and dimensions. We have derived similar bounds on coherence in these situations. Though we have omitted it in this paper due to space constraints, we have previously explored a way to use a variation of our method to conduct efficient searches over subsets of the representations of a finite cyclic group to quickly find group frames which achieve even lower coherence than those constructed in this paper (see \cite{Thill_ISIT2013}). In the sequel to this paper \cite{ThillGroupFourierFrames}, we will realize our method in a more general context, showing how to choose representations of a general group to construct group frames. We will develop a general framework which will tie all of our previous constructions together, and it will become apparent why our cyclic group construction extends so naturally to generalized dihedral groups. Furthermore, we will identify other groups for which our method produces frames with particularly low coherence, including certain other tight, equiangular frames. An interesting future direction would be to see whether the methods we have employed in this paper could be used to control the \textit{average coherence} of a frame, as defined in \cite{MixonBajwaCalderbank, BajwaCalderbankJafarpour}. This could allow us to construct frames satisfying the Strong Coherence Property \cite{MixonBajwaCalderbank, BajwaCalderbankJafarpour}, which have certain guaranteed performance in sparse signal processing.

\begin{appendices} 
\section{The Fourier Pairing of (\ref{eqn:alphaintermsofa}) and (\ref{eqn:aintermsofalpha}) for Cyclic Groups of Prime Order} 

We will now begin to develop the tools needed to prove Theorems \ref{thm:r2innerprods}, \ref{thm:r3cohbounds}, \ref{thm:coherenceupperbound} and \ref{thm:coherenceupperboundmodd}. We will explicitly prove Theorems \ref{thm:r2innerprods} and \ref{thm:r3cohbounds} and defer the proofs of Theorems \ref{thm:coherenceupperbound} and \ref{thm:coherenceupperboundmodd} to the sequel to this paper. Let us return to representations of the cyclic group $\Z/n\Z$, where $K = \{k_1, ..., k_m\}$ (not necessarily a group), $\U = \diag(\omega^{k_{1}},  \omega^{k_{2}}, ..., \omega^{k_{m}})$, with the powers $\omega^{k_i}$ distinct, and $\mathcal{U} = \{\U, \U^2, ..., \U^{n-1}, \U^n = \matr{I_m}\}$. As before, taking $\vv$ to be the normalized vector of all ones, $\U^\ell \vv = \begin{bmatrix} \omega^{k_{1} \ell}  &  \omega^{k_{2} \ell} & \hdots  &  \omega^{k_{m} \ell} \end{bmatrix}^T$. Then if we index the columns as $\ell = 0, 1, ..., n-1$, we have $\M$ as in (\ref{eqn:cyclicframemat}). The inner product associated to the element $\U^\ell$ takes the form $c_{\ell} := \frac{\vv^* \U^{\ell} \vv}{||\vv||_2^2} =  \frac{1}{m} \sum_{k \in K} \omega^{\ell k}$. We define $\alpha_\ell := |c_\ell|^2$ to be the squared norm of the $\ell^{th}$ inner product. If for any $t \in \Z/n\Z$ we define the set $A_t := \{(k_i, k_j) \in K \times K ~|~ k_i - k_j \equiv t \mod n \}$ with size $a_t := |A_t |$, then we have the Fourier pairing given by Equations (\ref{eqn:alphaintermsofa}) and (\ref{eqn:aintermsofalpha}).

Now consider the framework of Section \ref{sec:cyclicprimegrp} where $n$ is a prime, $m$ is a divisor of $n-1$, and $K$ is the unique cyclic subgroup of $(\Z/n\Z)^\times$. If $r = \frac{n-1}{m}$, then $K$ consists of the nonzero $r^{th}$ powers in $\Z/n\Z$. Let $x$ be a generator of $(\Z/n\Z)^\times$. Then the distinct cosets of $K$ in $(\Z/n\Z)^\times$ are $\{K, xK, x^2K, ..., x^{r-1}K\}$. If $\ell \in x^d K$, then we see that $c_\ell = c_{x^d}$ and hence $\alpha_\ell = \alpha_{x^d}$. Likewise, if $t \in x^d K$, it is not too difficult to see that we have a bijection 
\begin{align} 
\nonumber A_t \to A_{x^d}: \hspace{.1in} (k_i, k_j) \mapsto (x^d t^{-1} k_i, x^d t^{-1} k_j).
\end{align} 
It follows that 
\begin{align} 
a_t = a_{x^d} \text{ if } t \in x^d K. \label{eqn:atxdeqn} 
\end{align} 

It is straightforward to see from their definitions that $c_0 = \alpha_0 = 1$ and $a_0 = m$. With this in mind, we may write the condensed forms of (\ref{eqn:alphaintermsofa}) and (\ref{eqn:aintermsofalpha}): 
\begin{align} 
\nonumber \alpha_\ell  &= \frac{1}{m^2} \left(a_0 + \sum_{d = 0}^{r-1} a_{x^d} \sum_{k \in K} \omega^{x^d \ell k}\right), \\ 
\nonumber a_t  &= \frac{m^2}{n} \left( \alpha_0 + \sum_{d = 0}^{r-1} \alpha_{x^d} \sum_{k \in K} \omega^{-x^d t k}\right). 
\end{align} 
In particular, 
\begin{align}
\alpha_{x^{d'}} & = \frac{1}{m^2} \left(a_0 + \sum_{d = 0}^{r-1} a_{x^d} \sum_{k \in K} \omega^{x^{d+d'} k}\right) \\ 
&  = \frac{1}{m} \left(1 + \sum_{d = 0}^{r-1} a_{x^d} c_{x^{d+d'}} \right), \label{eqn:innerproductmagsq} \\ 
a_{x^{d'}} & = \frac{m^2}{n} \left( \alpha_0 + \sum_{d = 0}^{r-1} \alpha_{x^d} \sum_{k \in K} \omega^{-x^{d+d'} k}\right) \\ 
  & = \frac{m^2}{n} \left( 1 + m \sum_{d = 0}^{r-1} \alpha_{x^d} c^*_{x^{d + d'}} \right). 
 \end{align} 
 
On one final note, since the roots of unity sum to 0: 
\begin{align} 1 + m c_1 + m c_{x} + m c_{x^2} + ... + m c_{x^{r-1}} = 0. \label{eqn:ceqn} \end{align}

\section{$r = 2$, and Proof of Theorem \ref{thm:r2innerprods}} 
\label{sec:r2} 

As before, take $n$ to be a prime, $m$ a divisor of $n-1$, and $K = \{k_1, ..., k_m\}$ the unique multiplicative subgroup of $(\Z/n\Z)^\times$ of size $m$. Let us examine the case where $r := \frac{n-1}{m} = 2$. Fix a multiplicative generator $x$ for $(\Z/n\Z)^\times$. In this case, $K$ has two distinct cosets: $K$ and $xK$. Our frame will correspondingly have two distinct inner product values: $c_1 = \frac{1}{m} \sum_{k \in K} \omega^{k}$ and $c_x = \frac{1}{m} \sum_{k \in K} \omega^{x k}$. There are two equations of the form (\ref{eqn:innerproductmagsq}), 
\begin{align} 
\nonumber \alpha_1 & = \frac{1}{m} \left( 1 + a_1 c_1 + a_x c_x \right)~,~ \alpha_x  = \frac{1}{m} \left( 1 + a_1 c_x + a_x c_1 \right). 
\end{align} 
From (\ref{eqn:ceqn}), substituting $c_x = - \left(\frac{1}{m} + c_1 \right)$ gives us
\begin{align} 
\alpha_1 & = \frac{1}{m} \left( 1 - \frac{1}{m} a_x + (a_1 - a_x) c_1 \right),  \label{eqn:f1eqn} \\ 
\alpha_x & = \frac{1}{m} \left( 1 - \frac{1}{m} a_1 + (a_x - a_1) c_1 \right). \label{eqn:fxeqn} 
\end{align} 
From (\ref{eqn:f1eqn}) and (\ref{eqn:fxeqn}), we can see that since $\alpha_1, \alpha_x, a_1,$ and $a_x$ are real, then $c_1$ must be real as well (and thus so is $c_x$). This allows us to write 
\begin{align} 
\alpha_1 & = c_1^2, \hspace{.25in} \alpha_x = c_x^2 = \left(\frac{1}{m} + c_1 \right)^2. \label{eqn:alpha1c1}
\end{align}

\begin{lemma} 
Let $n$ be a prime, and $r$ and $m$ satisfy $r = \frac{n-1}{m} = 2$. Let $K$ be the unique subgroup of $(\Z/n\Z)^\times$ of size $m$. As before, let $a_t$ be the number of pairs $(k_1, k_2) \in K \times K$ such that $k_1 - k_2 = t$. Let $x$ be the multiplicative generator of $(\Z/n\Z)^\times$. Then 
\begin{itemize} 
\item If $n-1$ is divisible by 4, $a_1 = \frac{1}{2} (m-2)$ and $a_x = \frac{1}{2} m$. 
\item Otherwise, $a_1 = a_x = \frac{1}{2} (m-1), ~ (-1 \notin K)$. 
\end{itemize} 
\label{lem:asizer2lemma} 
\end{lemma}

\begin{proof} 
Let us first count the number of pairs $(k_1, k_2)$ such that $k_1 - k_2 \in K$, which will give us $\sum_{k \in K} a_k$. From (\ref{eqn:atxdeqn}), this is precisely equal to $m a_1$. Since $K$ is the group of nonzero squares in $\Z/p\Z$, we can write $k_1 = a^2$ and $k_2 = b^2$ for some choice of $a, b \in (\Z/p\Z)^\times$. If we let $x_1 = a - b$ and $x_2 = a + b$, then $k_1 - k_2 = (a - b)(a + b) = x_1 \cdot x_2$. Equivalently, we may write $$\begin{bmatrix} 1 & -1 \\ 1 & 1 \end{bmatrix} \begin{bmatrix} a \\ b \end{bmatrix} = \begin{bmatrix} x_1 \\ x_2 \end{bmatrix}. $$ We see that for any choice of the pair $(x_1, x_2)$, there is a unique pair $(a, b)$ which maps to it. Since we need to consider only pairs where $a$ and $b$ are nonzero, we must eliminate the cases where $x_1 = x_2$ (corresponding to when $b = 0$) and $x_1 = - x_2$ (corresponding to when $a = 0$). 

In order to have $x_1 \cdot x_2 \in K$, we must either have $x_1$ and $x_2$ both in $K$ or both in $xK$. If $-1 \in K$, a quick counting argument shows that there are $2m(m-2)$ valid choices for $(x_1, x_2)$ which satisfy $x_1 \cdot x_2 \in K$, each yielding a pair $(a, b)$ with $a$ and $b$ nonzero. But we are concerned only with their squares $a^2$ and $b^2$, so we can group these ordered pairs into sets of four, $\{ (\pm a, \pm b)\}$, and the number of distinct pairs $(a^2, b^2)$ with $a^2$ and $b^2$ nonzero and $a^2 - b^2 \in K$ is thus $$m a_1 = \frac{1}{4} (2 m(m-2)) = \frac{m}{2}(m-2), \text{ if } -1 \in K. $$

Likewise, $x_1 \cdot x_2 \in xK$ precisely when $x_1$ and $x_2$ are in opposite cosets of $K$. If this is true, and $-1 \in K$, then we cannot have $x_1 = x_2$ or $x_1 = -x_2$, since this would imply that $x_1$ and $x_2$ are in the same coset. Thus, any pair $(x_1, x_2)$ in either $K \times xK$ or $xK \times K$ will yield $x_1 \cdot x_2 \in xK$, so there are $2m^2$ possible pairs, each yielding a pair $(a, b)$. Again, we must divide by $4m$ to get the number of feasible pairs $(a^2, b^2)$ such that $a^2 - b^2 = x$, and we find that $$a_x = \frac{1}{2} m, ~ (-1 \in K). $$ 

If $-1 \notin K$, then the calculations for $a_1$ and $a_x$ change slightly: Now the condition $x_1 = -x_2$ implies that $x_1$ and $x_2$ are in opposite cosets of $K$. Thus, we have one extra case to consider when calculating $a_1$, and one less case when calculating $a_x$, so we find $$a_1 = a_x = \frac{1}{2} (m-1), ~ (-1 \notin K). $$ 

Note that $-1 \in K$, or rather -1 is a square modulo $n$, precisely when $(\Z/n\Z)^\times$ contains a fourth root of unity, and since $(\Z/n\Z)^\times$ is a cyclic multiplicative group of size $n-1$, this occurs precisely when $n-1$ is divisible by 4. 
\end{proof}

\begin{proof}
\textit{(Theorem \ref{thm:r2innerprods})} From (\ref{eqn:f1eqn}), (\ref{eqn:fxeqn}), (\ref{eqn:alpha1c1}), and Lemma \ref{lem:asizer2lemma}, we have that if $n-1$ is divisible by 4, $$c_1^2 = \frac{1}{m}\left(\frac{1}{2} - c_1 \right), $$ and making the substitution $c_1 = - \left( \frac{1}{m} - c_x \right)$ from (\ref{eqn:ceqn}) yields the same quadratic equation in $c_x$. Solving this reveals that $c_1$ and $c_x$ will take on the values $\frac{-1 \pm \sqrt{1 + 2m}}{2m}$, and the solution with the larger norm is $\frac{-1 - \sqrt{1 + 2m}}{2m}$, which indicates that the coherence is 
\begin{footnotesize} 
\begin{align} 
\nonumber \mu  = \left| \frac{-1 - \sqrt{1 + 2m}}{2m} \right| = \sqrt{\frac{n - m - \frac{1}{2}}{m(n-1)}} + \frac{1}{2m} \hspace{10pt} (n \equiv 1 \mod 4). 
\end{align}
\end{footnotesize}

On the other hand, if $n-1$ is not divisible by 4, then from Lemma \ref{lem:asizer2lemma} equations (\ref{eqn:fxeqn}) and (\ref{eqn:fxeqn}) become $$c_1^2  = c_x^2 = \frac{1}{m} \left(\frac{1}{2} + \frac{1}{2m}\right), $$ so this gives us coherence 
\begin{align} 
\nonumber \mu = \sqrt{\frac{1}{m} \left(\frac{1}{2} + \frac{1}{2m}\right)} = \sqrt{\frac{n-m}{m(n-1)}} \hspace{10pt} (n \not \equiv 1 \mod 4). 
\end{align} 
\end{proof}

\section{$r = 3$, and Proof of Theorem \ref{thm:r3cohbounds}} 
\label{sec:r3} 

Take $n$ to be a prime, $m$ a divisor of $n-1$, and $K = \{k_1, ..., k_m\}$ the unique subgroup of $(\Z/n\Z)^\times$ of size $m$. We now consider the case where $r = \frac{n-1}{m} = 3$, so that if $x$ is a generator of $(\Z/n\Z)^\times$, then $K$ is cyclically generated by $x^3$, and consists of the cubes of all the nonzero integers modulo $n$. In this case our distinct inner products will be $c_1, c_x$, and $c_{x^2}$, with corresponding squared norms $\alpha_1$, $\alpha_x$, and $\alpha_{x^2}$. Our goal in this section will be to prove Theorem \ref{thm:r3cohbounds}.

We first make the following remark: 

\begin{lemma} 
Let $n$ be a prime, $\omega = e^\frac{2 \pi i}{n}$, and $r$ and $m$ satisfy $r = \frac{n-1}{m} = 3$. If we take $K$ to be the unique subgroup of $(\Z/n\Z)^\times$ of size $m$, then the inner product values $c_\ell = \frac{1}{m} \sum_{k \in K} \omega^{\ell k}$ are all real. 
\label{lem:r3realinnerprods} 
\end{lemma} 

\begin{proof} 
$K$ is the set of cubes in $(\Z/n\Z)^\times$, and since $-1$ is its own cube it will lie in $K$. Multiplication by $-1$ will therefore permute the elements of $K$, so we have 
\begin{align} 
c_{\ell}^* & = \left(\frac{1}{m} \sum_{k \in K} \omega^{\ell k}\right)^* = \frac{1}{m} \sum_{k \in K} \omega^{-\ell k} = c_{\ell}. 
\end{align} 

\end{proof}

We begin by making the following definition: 

\begin{definition} 
For any two cosets $t_1 K $ and $t_2 K$, we define the \textit{translation degree from $t_1 K$ to $t_2 K$}, to be the quantity
\begin{align} 
\nonumber a_{t_1 K, t_2 K} &:= | (1 + t_1 K) \cap t_2 K| = \#\{ \alpha \in t_1 K~|~1+\alpha \in t_2 K\}. 
\end{align} 
Similarly, for any coset $tK$, define the \textit{translation degree from $t K$ to 0} to be the quantity 
\begin{align} 
\nonumber a_{tK, 0} := |(1+t K) \cap \{0\}| = 
\begin{cases} 
1 & \text{ if } -1 \in tK, \\ 
0 & \text{ otherwise. } 
\end{cases} 
\end{align} 

\end{definition} 

We can express our previously defined values $a_t$ in terms of the translation degrees as follows: 

\begin{lemma} 
Let $n$ be a prime, and $m$ and $r$ satisfy $r = \frac{n-1}{m} = 3$. Let $K$ be the unique subgroup of $(\Z/n\Z)^\times$ of size $m$. Define $a_{t} = \#\{(k_1, k_2) \in K \times K~|~k_1 - k_2 \equiv t \mod n\}$ Then $a_{t} = a_{K, t K}$. 
\label{lem:linktotransdeg} 
\end{lemma} 

\begin{proof} 
For every pair $(k_1, k_2) \in K \times K$ we have that $k_1 - k_2 \in tK$ if and only if $1 - k_2 k_1^{-1} \in tK$. There are $m a_t$ such pairs in total ($a_t$ pairs for every element in $tK$). Note that $- k_2 k_1^{-1} \in K$, since $-1 \in K$. If we select any of the $m$ candidates for $k_1 \in K$, then there are $a_{K, tK}$ choices for $k_2$ that will satisfy this requirement. Thus, we have $m a_t = m a_{K, tK}$, and the result follows. 
\end{proof}

Some other facts about translation degrees: 

\begin{lemma} 
Let $n$ be a prime, $m$ a divisor of $n-1$ such that $\frac{n-1}{m} = 3$, and $K$ the unique subgroup of $(\Z/n\Z)^\times$ of size $m$. Then $a_{t_1K, t_2K} = a_{t_2 K, t_1 K}$ for all $t_1, t_2 \in \Z/n\Z$.   
\label{lem:transdegcommutingcosets}
\end{lemma} 

\begin{proof} Suppose $b_1 \in t_1K$ such that $1+b_1 = b_2 \in t_2 K$. Then $1 - b_2 = -b_1$, with $-b_2 \in t_2K$ and $-b_1 \in t_1 K$ (since $-1$ is a cube and is thus in $K$). In fact, we see that we have a bijection between the sets $\{(b_1, b_2) \in t_1 K \times t_2 K ~|~ 1 + b_1 = b_2\}$ and $\{(c_1, c_2) \in t_2 K \times t_1  K ~|~ 1 + c_1 = c_2\}$ which sends $(b_1, b_2) \mapsto (c_1, c_2) = (-b_2, -b_1)$. This gives us $a_{t_1K, t_2K} = a_{t_2 K, t_1 K}$. 
\end{proof} 

\begin{lemma} 
Let $n$ be a prime, $m$ a divisor of $n-1$ such that $r = \frac{n-1}{m}$, and $K$ the unique subgroup of $(\Z/n\Z)^\times$ of size $m$. If $x$ is the multiplicative generator of $(\Z/n\Z)^\times$, then $a_{x^i K, x^j K} = a_{x^{r - i} K, x^{r - i + j} K}$.  
\label{lem:transdegshiftingcosets} 
\end{lemma} 

\begin{proof} 
Let $a \in K$ such that $1 + x^i a = x^j b$, with $b \in K$. Then multiplying both sides of this equation by $x^{r - i}$, we get $x^{r-i} + x^r a = x^{r - i + j} b$. Note that $x^r a \in K$. Now, multiplying both sides of this equation by $(x^r a)^{-1} \in K$, we obtain $1 + x^{r-i}(x^r a)^{-1} = x^{r - i + j} b (x^r a)^{-1}$, where $x^{r-i}(x^r a)^{-1} \in x^{r-i}K$ and $x^{r - i + j} b (x^r a)^{-1} \in x^{r - i + j}K $. We see that we in fact have a bijection between the sets $\{(x^i a, x^j b) \in x^i K \times x^j K~|~ 1+x^i a = x^j b\}$ and $\{(x^{r-i} c, x^{r-i+j}d) \in x^{r-i} K \times x^{r-i+j} K~|~ 1+x^{r-i} c =x^{r-i+j}d\}$ which sends $(x^i a, x^j b) \mapsto (x^{r-i}(x^r a)^{-1}, x^{r - i + j} b (x^r a)^{-1}). $ 
\end{proof}

\begin{lemma} 
Let $n$ be a prime, $m$ a divisor of $n-1$ such that $r = \frac{n-1}{m} $, and $K$ the unique subgroup of $(\Z/n\Z)^\times$ of size $m$. Set $G = (\Z/n\Z)^\times$, with multiplicative generator $x$. For any coset $t_0 K$, we have $a_{t_0 K, 0} + \sum_{i = 1}^r a_{t_0 K, x^i K} = |t_0 K|$.  
\label{lem:sumoftransdegs} 
\end{lemma} 

\begin{proof} 
This simply follows from the observation that any element of $t_0 K$, when translated by 1, must be sent to either 0 or exactly one of the cosets $x^i K \in G/K$. 
\end{proof}

\begin{lemma} 
Let $n$ be a prime, and $m$ a divisor of $n-1$ such that $r := \frac{n-1}{m} = 3$. Take $K$ to be the unique subgroup of $(\Z/n\Z)^\times$ of size $m$, and $x$ a multiplicative generator for $(\Z/n\Z)^\times$. Then 
\begin{align} 
a_{xK, x^2 K} - a_{K, K} & = 1. \label{eqn:transdegeq1}
\end{align} 
\label{lem:transdegeqnlemmar3} 
\end{lemma} 
\vspace{-15pt} 

\begin{proof} 
We prove this by counting the size of the set 
\begin{align} 
\nonumber A_K := \{(k_1, k_2) \in K \times K~|~ k_1 - k_2 \in K \}
\end{align} 
in two ways. First, using Equation (\ref{eqn:atxdeqn}), we can simply count the elements in this set as 
\begin{align} 
|A_K| = \sum_{k \in K} a_K = m a_1. \label{eqn:AKformula1} 
\end{align} 

Alternatively, we note that when $r = 3$, the difference between any two elements in $K$ takes the form $$a^3 - b^3 =(a - b)(a - \zeta b)(a - \zeta^2 b), $$ where $\zeta$ is a primitive third root of unity, and $a$ and $b$ are nonzero. Let us define 
\begin{align} 
x_1 := a - b, ~x_2 := a - \zeta b, ~ x_3 := a - \zeta^2 b. \label{eqn:x1x2x3} 
\end{align} 
We can express this using matrices as $$\begin{bmatrix} 1 & -1 \\ 1 & -\zeta \\ 1 & - \zeta^2 \end{bmatrix} \cdot \begin{bmatrix} a \\ b \end{bmatrix} = \begin{bmatrix} x_1 \\ x_2 \\ x_3 \end{bmatrix}. $$ In this form we can see that $a$ and $b$, and $x_3$ are uniquely determined by $x_1$ and $x_2$. In particular, 
\begin{align} 
x_3 = -\zeta(x_1 + \zeta x_2). 
\end{align}

Now, if $a^3 - b^3 \in K$, then we have the following possibilities for which cosets of $K$ $x_1, x_2$, and $x_3$ must belong to (up to a permutation of the cosets): 

\begin{table}[ht]
\caption{$a^3 - b^3 \in K$} 
\centering 
\begin{tabular}{c c c c}
\hline 

$x_1$ & $x_2$ & $x_3$ & Multiplicity  \\ 
\hline 
$K$ & $K$ & $K$ & 1 \\ 
$xK$ & $xK$ & $xK$ & 1 \\ 
$x^2K$ & $x^2K$ & $x^2K$ & 1 \\ 
$K$ & $xK$ & $x^2K$ & 6  \\ 
\hline 
\end{tabular} 
\end{table}  

The last case is representative of six possible cases which we obtain by permuting the order of the cosets (thus it has ``multiplicity 6"). In short, we must have $x_1, x_2$, and $x_3$ all in the same coset, or all in different cosets of $K$ in order to have $a^3 - b^3 \in K$. Let us attempt to count the quantity $$\#\{(x_1, x_2) \in K \times K ~|~ x_3 = -\zeta(x_1 + \zeta x_2) \in K, ~a \ne 0, b \ne 0 \}. $$ Since $x$ generates $(\Z/n\Z)^\times \cong \mathbb{F}_n^\times$, and $r$ divides $n-1$, the order of this group, then any $r^{th}$ root of unity will be contained in $\mathbb{F}_n^\times$, so $\zeta$ will lie in one of the cosets of $K$. 

We will first consider the case where $\zeta \in K$. Since $r = 3$, $-1 \in K$, so $- \zeta \in K$. Thus, the condition that $-\zeta(x_1 + \zeta x_2) \in K$ is equivalent to the condition that $x_1 + \zeta x_2 \in K \iff 1 + \zeta x_2 x_1^{-1} \in K$. If we fix $x_1$ to be any one of the $m$ elements in $K$, we have exactly $a_{K,K}$ choices for $x_2$ which satisfy this condition (for every $k \in K$ such that $1 + k \in K$, simply set $x_2 = k x_1 \zeta^{-1}$). This gives us a total of $m a_{K, K}$ ordered pairs $(x_1, x_2) \in K \times K$, each corresponding to a unique pair $(a, b)$ with $a^3 - b^3 \in K$. But we must rule out those which have either $a$ or $b$ equal to zero. If $a = 0$, then any choice of $b \in K$ will satisfy that all the $x_i$ are in $K$. Likewise, if $b = 0$, then any choice of $a \in K$ will do the same. Thus, there are $2m$ cases to eliminate, so 
\begin{align} 
\nonumber \#\{(x_1, x_2) \in K \times K ~|~ x_3  \in K, ~a \ne 0, b \ne 0 \} \\ 
= m a_{K, K} - 2m. \label{eqn:allinK} 
\end{align} 

By mimicking these calculations, it is not too difficult to see that we also have 
\begin{small} 
\begin{align} 
&~\#\{(x_1, x_2) \in xK \times xK ~|~ x_3  \in xK, ~a \ne 0, b \ne 0 \}  \label{eqn:allinxK}  \\ 
=& ~\#\{(x_1, x_2) \in x^2 K \times x^2 K ~|~ x_3  \in x^2 K, ~a \ne 0, b \ne 0 \}  \label{eqn:allinx2K} \\ 
=& ~ m a_{K, K} - 2m. 
\end{align} 
\end{small} 

Now consider the case where $x_1$, $x_2$, and $x_3$ are each in different cosets of $K$. We see that this rules out the case where either $a$ or $b$ is zero, since this would force all the $x_i$ to be in the same coset. Suppose $x_1 \in K$, $x_2 \in xK$, and $x_3 \in x^2 K$. Since $x_3 = - \zeta (x_1 + \zeta x_2)$, we must have $1 + \zeta x_2 x_1^{-1} \in x^2 K$, where we note that $x_2 x_1^{-1} \in xK$. For any fixed $x_1 \in K$, there are $a_{xK, x^2 K}$ choices for $x_2$ that satisfy this constraint. Thus, we arrive at 
\begin{align}
\nonumber &\#\{(x_1, x_2) \in K \times xK ~|~ x_3  \in x^2 K, ~a \ne 0, b \ne 0 \} \\ 
&= m a_{xK, x^2 K}. \label{eqn:diffKcosets} 
\end{align} 
With a little work exploiting Lemma \ref{lem:transdegcommutingcosets}, we see that we will arrive at the same result for any of the six permutations of the cosets corresponding to $x_1$, $x_2$, and $x_3$.

We comment that for any ordered pair $(k_1, k_2) \in K \times K$ such that $k_1 - k_2 \in K$ the nine pairs $(a, b) = (\zeta^{n_1} k_1^{1/3}, \zeta^{n_2} k_2^{1/3})$, for $n_1$ and $n_2$ ranging independently between 0 and 2, will all satisfy $(a^3, b^3) = (k_1, k_2)$. Thus, in counting the size of $A_K$, we will have to add up our previous quantities from (\ref{eqn:allinK}), (\ref{eqn:allinxK}), (\ref{eqn:allinx2K}), and (\ref{eqn:diffKcosets}) (with multiplicities) and then divide by 9. This gives us 
\begin{align} 
|A_K| = \frac{1}{9} \left( 3(m a_{K, K} - 2m) + 6m a_{xK, x^2 K}\right). \label{eqn:AKformula2} 
\end{align} 
Finally, combining (\ref{eqn:AKformula1}) and (\ref{eqn:AKformula2}), and using Lemma \ref{lem:linktotransdeg} to make the substitution $a_1 = a_{K, K}$, we obtain the result for the case where $\zeta \in K$. 

For the case where $\zeta \notin K$ we can verify that the relation does in fact still hold. It suffices to prove the result for when $\zeta \in xK$, for the result will also hold when $\zeta \in x^2 K$ due to the interchangeability of $xK$ and $x^2K$ which arises from both being multiplicative generators of $G/K$. In this case, we can show using similar counting arguments as before that for $d = 0, 1, 2,$ 
\begin{align} 
\nonumber &\#\{(x_1, x_2) \in x^d K \times x^{d} K~|~x_3 \in x^{d} K, a \ne 0, b \ne 0\} \\ 
= & ~ m a_{xK, x^2 K} - m, \\ 
\nonumber &\#\{(x_1, x_2) \in x^d K \times x^{d+1} K~|~x_3 \in x^{d+2} K, a \ne 0, b \ne 0\} \\ 
 = & ~ m a_{x^2 K, x K} - m, \\ 
\nonumber &\# \{(x_1, x_2) \in x^d K \times x^{d+2} K~|~x_3 \in x^{d+1} K, a \ne 0, b \ne 0\} \\
 = & ~ m a_{K, K}.
\end{align} 
Summing these values up for $d = 1, 2, 3$, and again dividing by 9 and equating the value to (\ref{eqn:AKformula1}), we obtain 
\begin{align} 
\nonumber m a_1 & = \frac{1}{9}  \left( 3( m a_{xK, x^2K} - m) \right) \\ 
& + \frac{1}{9} \left( 3( m a_{x^2K, xK} - m) +3 m a_{K, K} \right), 
\end{align} 
which after substituting $a_1 = a_{K, K}$ and $a_{x^2K, xK} = a_{xK, x^2K}$ (from Lemmas \ref{lem:linktotransdeg} and \ref{lem:transdegcommutingcosets}) reduces to the desired relation $a_{xK, x^2K} - a_{K, K} = 1$. 
\end{proof}

\begin{lemma} 
Let $n$ be a prime, $m$ a divisor of $n-1$ such that $\frac{n-1}{m} = 3$, and $K$ the subgroup of $(\Z/n\Z)^\times$ of size $m$. Then if $x$ is a multiplicative generator for $(\Z/n\Z)^\times$, $\omega = e^\frac{2 \pi i}{n}$, and $c_{\ell} = \frac{1}{m} \sum_{k \in K} \omega^{\ell k}$ is the inner product value corresponding to $\ell \in \Z/n\Z$, then 
\begin{align} 
\cc \cc^* = \frac{1}{m} [ I - \diag(\cc)  + P(I + A)C], \label{eqn:cmatrixeqn} 
\end{align}
where $\cc = [c_1, c_x, c_{x^2}]^T$, $I$ is the $3 \times 3$ identity matrix, and $$A = \begin{bmatrix} a_1 & a_{x^2} & a_x \\ a_x & a_1 & a_{x^2} \\ a_{x^2} & a_x & a_1 \end{bmatrix} , ~ C = \begin{bmatrix} c_1 & c_{x^2} & c_x \\ c_x & c_1 & c_{x^2} \\ c_{x^2} & c_x & c_1 \end{bmatrix}, $$ $$ P = \begin{bmatrix} 1 & 0 & 0 \\ 0 & 0 & 1 \\ 0 & 1 & 0 \end{bmatrix}. $$

\end{lemma} 

\begin{proof} 
The terms of $\cc \cc^*$ will take the form 
\begin{align} 
c_{x^i} c_{x^j}^* = \frac{1}{m^2} \sum_{(k_1, k_2) \in K \times K } \omega^{x^i k_1 - x^j k_2}. \label{eqn:unrefinedcterms} 
\end{align} 
If $i \ne j$, note that, $x^i k_1 - x^j k_2 \in x^d K$ if and only if $1 - x^{j- i} k_2 k_1^{-1} \in x^{d - i} K$, and there are $m a_{x^{j-i}K, x^{d-i}K}$ choices for $(k_1, k_2)$ that satisfy this. Thus, we obtain 
\begin{align} 
c_{x^i} c_{x^j}^* & = \frac{1}{m} \sum_{d = 0}^{r-1} a_{x^{j-i}K, x^{d-i}K} c_{x^d}. ~ (i \ne j) \label{eqn:crossterms}
\end{align} 
If $i = j = d'$, (\ref{eqn:unrefinedcterms}) becomes $\frac{1}{m^2} \sum_{(k_1, k_2) \in K \times K } \omega^{x^{d'} (k_1 - k_2)}$. Separating the terms where $k_1 = k_2$, we can apply the same reasoning as above and use Lemma \ref{lem:linktotransdeg} to obtain 
\begin{align} 
|c_{x^{d'}}|^2 = \frac{1}{m} \left(1 + \sum_{d = 0}^{r-1} a_{x^d} c_{x^{d+d'}}\right). \label{eqn:diagterms} 
\end{align} 
Equation (\ref{eqn:cmatrixeqn}) can now be verified from (\ref{eqn:crossterms}) and (\ref{eqn:diagterms}) using Lemmas \ref{lem:linktotransdeg}, \ref{lem:transdegcommutingcosets}, \ref{lem:transdegshiftingcosets}, \ref{lem:sumoftransdegs}, and \ref{lem:transdegeqnlemmar3}. 
\end{proof}

\begin{proof} 
\textit{(Theorem \ref{thm:r3cohbounds})} Notice in (\ref{eqn:cmatrixeqn}) that $A$ and $C$ are circulant matrices (as is $I + A$), and hence they can be diagonalized by Fourier matrices. Let $\gamma = e^{2 \pi i /3}$ and $$F = \begin{bmatrix} 1 & 1 & 1 \\ 1 & \gamma & \gamma^2 \\ 1 & \gamma^2 & \gamma^4 \end{bmatrix} = \begin{bmatrix} 1 & 1 & 1 \\ 1 & \gamma & \gamma^{-1} \\ 1 & \gamma^{-1} & \gamma \end{bmatrix}, $$ so that $\frac{1}{\sqrt{3}} F$ is the $3 \times 3$ discrete Fourier matrix. We first note that the matrix $P$ from above is simply $\frac{1}{3} F^2 = \frac{1}{3}F^{*2}$. Now it is easy to verify that since $\cc$ has real components by Lemma \ref{lem:r3realinnerprods}, then if we write $F \cc = [w_1, w_2, w_3]^T$, then we have that $w_1$ is real and $w_2 = w_3^*$. So we may write $w_1 = \alpha$, $w_2 = \beta e^{i \theta}$, and $w_3 = \beta e^{-i \theta}$, where $\alpha$ and $\beta$ are real and $\beta$ is nonnegative. If we let $\va = [a_1, a_x, a_{x^2}]^T$, then we can easily verify that by pre-multiplying Equation (\ref{eqn:cmatrixeqn}) by $F$ and post-multiplying by $F^*$, noting that $FF^* = 3 I$, $FPF^* = 3P$, $FCF^* = 3 \diag(F \cc)$ and $FAF^* = \diag(F \va)$, we can rewrite it as 
\begin{align} 
\nonumber (F \cc)(F \cc)^* & =  \frac{1}{m} [3 I - F \diag(\cc) F^* \\ 
 & + 27 P (I + \diag(F \va)) \diag(F \cc)]. \label{eqn:cmatrixeqn2}
\end{align} 

One can further check that $F \diag(\cc) F^*$ is circulant with first column $F \cc$, and if we write $F \va = [y_1, y_2, y_3]^T$, then (\ref{eqn:cmatrixeqn2}) becomes 
\begin{align} 
\begin{bmatrix} w_1 \\ w_2 \\ w_3 \end{bmatrix} [w_1^*, w_2^*, w_3^*]  = & \frac{1}{m} [ \begin{bmatrix} 3 & 0 & 0 \\ 0 & 3 & 0\\ 0 & 0 & 3 \end{bmatrix}  - \begin{bmatrix} w_1 & w_3 & w_2 \\ w_2 & w_1 & w_3\\ w_3 & w_2 & w_1 \end{bmatrix} \label{eqn:cmatrixeqn3}  
\end{align} 
$$ + 27 \begin{bmatrix} (1+y_1)w_1 & 0 & 0 \\ 0 & 0 & (1+y_3) w_3 \\ 0 & (1+y_2) w_2 & 0 \end{bmatrix} ]. $$
If we consider only the coordinates of the above matrices which do not involve $y_1, y_2$ or $y_3$, then after substituting $w_1 = \alpha$, $w_2 = \beta e^{j \theta}$ and $w_3 = \beta e^{-j \theta}$, we can solve the resulting equations to obtain the relations 
\begin{equation} 
\alpha  = - \frac{1}{m}, \hspace{.25in} \beta = \sqrt{\frac{1}{m} \left(3 + \frac{1}{m} \right)}. \label{eqn:alphaeqn}
\end{equation}

We can use these to bound the coherence as follows: 
\begin{align} 
\begin{bmatrix} c_1 \\ c_x \\ c_{x^2} \end{bmatrix} & = F^{-1} \begin{bmatrix} \alpha \\ \beta e^{j \theta} \\ \beta e^{-j \theta} \end{bmatrix} = \frac{1}{3} \begin{bmatrix} \alpha + 2 \beta \cos(\theta) \\ \alpha + 2 \beta \cos(\theta - \frac{2\pi}{3}) \\ \alpha + 2 \beta \cos(\theta + \frac{2 \pi}{3}) \end{bmatrix}.  
\end{align}

\begin{align} \min_\theta  \max\{|c_1|, |c_x|, |c_{x^2}|\} \le \mu \le \max_\theta  \max\{|c_1|, |c_x|, |c_{x^2}|\} 
\end{align}

From (\ref{eqn:alphaeqn}), we know that $\alpha$ is negative, and $\beta$ is positive by definition. Since $|\alpha| < |\beta|$, then by inspection we have  
\begin{align} 
\max_\theta \max\{|c_1|, |c_x|, |c_{x^2}|\} & = \frac{1}{3} |\alpha + 2 \beta (-1)| \\ 
& = \frac{1}{3} \left(2\sqrt{\frac{1}{m} \left(3 + \frac{1}{m} \right)} + \frac{1}{m}\right). 
\end{align}  

This gives us our upper bound. Asymptotically, we can ignore the term $\alpha = - \frac{1}{m}$ in our expressions for $c_1, c_x$, and $c_{x^2}$, and if we do so, we find that $$\arg \min_\theta  \max\{|c_1|, |c_x|, |c_{x^2}|\} \approx \frac{\pi}{2}, $$ which follows from noting that since $|c_1|, |c_x|$, and $|c_{x^2}|$ are continuous functions of $\theta$, the smallest value of their maximum must occur when two of them are set equal to each other (in this case, when $|c_x| = |c_{x^2}|$, so that asymptotically $| \cos(\theta + \frac{2 \pi}{3})| = | \cos(\theta - \frac{2 \pi}{3})|$). Substituting $\frac{\pi}{2}$ for $\theta$ gives us our (asymptotic) lower bound on $\mu$: 
\begin{align} 
\min_\theta  \max\{|c_1|, |c_x|, |c_{x^2}|\} \approx \frac{1}{\sqrt{m}}. 
\end{align} 
We easily verify that this is greater than the Welch bound, which in this case becomes $$\sqrt{\frac{n-m}{m(n-1)}} = \sqrt{ \frac{2}{3m} + \frac{1}{3m^2}}. $$ 

\end{proof}

\section{Proof of Theorems \ref{thm:coherenceupperbound} and \ref{thm:coherenceupperboundmodd}} 
\label{sec:generalr} 

We now exploit the tools developed in the preceding two appendix sections to generalize our bound from Theorem \ref{thm:r3cohbounds} to general values of $r$. We first examine the Fourier transform of our vector $[c_1,  c_x, ...,  c_{x^{r-1}}]^T$ of our inner product values.

\begin{lemma} Let $n$ be a prime, $m$ a divisor of $n-1$, and $r := \frac{n-1}{m}$. Let $x$ be a multiplicative generator of the cyclic group $(\Z/n\Z)^\times$, $K$ the unique subgroup of $(\Z/n\Z)^\times$ of size $m$, and $c_{x^d}$ be the inner product value $\frac{1}{m} \sum_{k \in K} \omega^{x^d \cdot k}$ where $d \in \{0, ..., r-1\}$ and $\omega = e^{\frac{2 \pi i}{n}}$. Finally, let $\cc = [c_1, c_x, c_{x^2}, ..., c_{x^{r-1}}]^T$, and let $F$ be the scaled $r \times r$ Fourier matrix with entries defined by $F_{ij} = \gamma^{(i - 1) (j- 1)}$, where $\gamma = e^{\frac{2 \pi i}{r}}$. 
Then, if we let $\w := [w_1, ..., w_r]^T = F \cc$ so that $w_{d + 1} = \sum_{t = 0}^{r-1} \gamma^{t d} c_{x^{t}}$ for $d = 0, 1, ..., r-1$, we have 

\begin{align} 
w_1 &= - \frac{1}{m}, &  \label{eqn:wvec1} \\ 
|w_i| &= \sqrt{\frac{1}{m} \left(r + \frac{1}{m} \right)}, & i \ne 1. \label{eqn:wvecnot1} 
\end{align} 

\label{lem:wvec} 
\end{lemma}

\begin{proof} 
Note that (\ref{eqn:wvec1}) follows from (\ref{eqn:ceqn}) since $w_1 = \sum_{t = 0}^{r-1} c_{x^t} = - \frac{1}{m}$. 

Now, for any $d \in \{0, ..., r-1\}$, 

\begin{align} 
|w_{d+1}|^2 & = \left(\sum_{t = 0}^{\kappa-1} \gamma^{t d} c_{x^{t}} \right) \left( \sum_{\ell = 0}^{\kappa-1} \gamma^{\ell d} c_{x^{\ell}} \right)^* \\ 
& = \sum_{t = 0}^{\kappa-1} \sum_{\ell = 0}^{\kappa-1} \gamma^{(t - \ell) d} c_{x^t} c_{-x^\ell}\\
& = \sum_{s = 0}^{\kappa-1}  \gamma^{s d} \sum_{\ell = 0}^{\kappa-1} c_{x^{s+\ell}} c_{- x^\ell}  \label{eqn:wnormsq} 
\end{align} 

Also, we have 

\begin{align} 
m^2 c_{x^{s + \ell}} c_{- x^{\ell}} & = \left(\sum_{k \in K} \omega^{x^{s+ \ell} k}\right) \left(\sum_{k' \in K} \omega^{-x^\ell k'} \right) \\ 
& = \sum_{k, k' \in K} \omega^{-x^\ell k' (1 - x^{s} k k'^{-1})} \\ 
& = \sum_{t = 0}^{r - 1} \sum_{\substack{\{k', k'' \in K ~: \\ 1 - x^s k'' \in x^t K\}}} \omega^{-x^\ell k' (1 - x^{s} k'')} \\ 
\nonumber & \hspace{40pt}+ \sum_{\substack{\{k', k'' \in K ~: \\ 1 - x^s k'' = 0\}}} 1 \\ 
& = \sum_{t = 0}^{r-1} a_{-x^s K, x^t K} \left( \sum_{k''' \in K} \omega^{-x^t x^\ell k'''} \right) \\ 
& \hspace{40pt} + \sum_{k' \in K}  a_{-x^s K, 0} \\ 
& = m \sum_{t = 0}^{r-1} a_{-x^s K, x^t K} \cdot c_{-x^{t +\ell}} + m a_{-x^s K, 0} 
\end{align}

Dividing both sides by $m^2$ and substituting into (\ref{eqn:wnormsq}), we obtain: 

\begin{align} 
|w_{d+1}|^2 & = \sum_{s = 0}^{r-1}  \gamma^{s d} \sum_{\ell = 0}^{r-1} \left( \frac{1}{m} \left(\sum_{t = 0}^{r-1} a_{-x^s K, x^t K} \cdot c_{-x^{t +\ell}} + a_{-x^s K, 0} \right) \right) \\ 
& = \sum_{s = 0}^{r-1}  \gamma^{s d}   \frac{1}{m} \left(\sum_{t = 0}^{r-1} a_{-x^s K, x^t K} \sum_{\ell = 0}^{r-1} c_{-x^{t +\ell}} + \sum_{\ell = 0}^{r-1} a_{-x^s K, 0} \right) \\ 
& = \sum_{s = 0}^{r-1}  \gamma^{s d}   \frac{1}{m} \left(\sum_{t = 0}^{r-1} a_{-x^s K, x^t K} \left( -\frac{1}{m} \right) + r a_{-x^s K, 0} \right) \label{eqn:wnormsqeq2} \\ 
& = - \frac{1}{m^2} \sum_{s = 0}^{r-1}  \gamma^{s d} \left(m - a_{-x^s K, 0} \right) + \frac{r}{m} \sum_{s = 0}^{r-1}  \gamma^{s d} a_{-x^s K, 0} \label{eqn:wnormsqeq4}
\end{align}

\noindent
where (\ref{eqn:wnormsqeq2}) follows from Equation (\ref{eqn:ceqn}), and (\ref{eqn:wnormsqeq4}) follows from Lemma \ref{lem:sumoftransdegs}. Since $a_{-x^s K, 0}$ is equal to $1$ if $s = 0$ and equal to 0 otherwise, (\ref{eqn:wnormsqeq4}) becomes 
\begin{align} 
|w_{d+1}|^2 & = -\frac{1}{m^2} \left( (m-1) + m \sum_{s = 1}^{r-1}  \gamma^{s d} \right) + \frac{r}{m}.  \label{eqn:wnormsqeq5} 
\end{align} 

When $d \ne 0$ we have $\sum_{s = 1}^{r-1}  \gamma^{s d} = -1$, and after rearranging terms we obtain 
\begin{align} 
|w_{d+1}|^2 & = \frac{1}{m} \left( r + \frac{1}{m} \right),  
\end{align} 
giving us (\ref{eqn:wvecnot1}). 
\end{proof}

The proof of Theorem \ref{thm:coherenceupperbound} follows immediately from this lemma: 

\begin{proof} 
\textit{(Theorem \ref{thm:coherenceupperbound})} Using the notation of Lemma \ref{lem:wvec}, write $\mathbf{c} = \frac{1}{r} F^* \mathbf{w}$. Then 
\begin{align} 
|c_{x^d}| & = \frac{1}{r} \left| \sum_{j = 1}^r \gamma^{d (j-1)} w_j \right| \\ 
& \le \frac{1}{r}  \sum_{j = 1}^r |w_j| \label{eqn:cohprooftriineq} \\  
& = \frac{1}{r} \left( (r-1) \sqrt{\frac{1}{m} \left(r - \frac{1}{m}\right)} + \frac{1}{m} \right), \label{eqn:cohproofwveclemma} 
\end{align} 
where (\ref{eqn:cohproofwveclemma}) follows from Lemma \ref{lem:wvec}. Since the coherence is equal to the largest value among the $|c_{x^d}|$, $d = 0, ..., r - 1$, we are done. 
\end{proof}

Now, toward proving Theorem \ref{thm:coherenceupperboundmodd}, we present the following classification of when $|K|$ is even or odd: 

\begin{lemma}
Let $n$ be a prime, $m$ a divisor of $n - 1$, and $r := \frac{n - 1}{m}$. 
Let $x$ be a generator for the cyclic multiplicative group $(\Z/n\Z)^\times$, and let $K$ be the unique subgroup of $(\Z/n\Z)^\times$ of size $m$. 
Then $-1 \in K$ if and only if either $n$ or $m$ is even. If $n$ and $m$ are both odd, then $r$ is even and $-1 \in x^\frac{r}{2} A$. 
\label{lem:mevennegoneinA}
\end{lemma}

\begin{proof} 
If $n$ is even, that is $n = 2$, then $-1 \equiv 1$ in $\Z/n\Z$, so trivially $-1 \in K$. If $n$ is odd, then the size $m$ of $K$ is even if and only if $K$ contains the unique cyclic subgroup of size $2$, which is $\{\pm 1\}$. 

If both $m$ and $n$ are odd, then $n - 1$ must be even, hence so is $r = \frac{n- 1}{m}$. By the argument above, $-1 \notin K$. In this case, since $(-1)^2 = 1 \in K$ (and noting that $K = x^r K$), we must have $-1 \in x^\frac{r}{2} K$. 
\end{proof}

\begin{lemma} 
Let $n$, $m$, $r$, $x$ and $K$ be defined as in Lemma \ref{lem:mevennegoneinA} and $\mathbf{c} = [c_1, c_{x}, ..., c_{x^{r-1}}]^T$ and $\w = [w_1, ..., w_r]^T$ be defined as in Lemma \ref{lem:wvec}. 
If either $n$ or $m$ is even ($-1 \in K$) then for any $d = 0, 1, ..., r-1$, we have $c_{x^d} = c^*_{x^d}$, and for any $i = 2, 3, ..., r$ we have $w^*_i = w_{r-i + 2}$. If $n$ and $m$ are both odd ($-1 \in x^\frac{r}{2} K$), then  $c_{x^d} = c^*_{x^{d + r/2}}$ and $w^*_i = (-1)^{i - 1} w_{r-i + 2}$. 
\label{lem:cstarwstar} 
\end{lemma}

\begin{proof} 
As usual, set $\omega = e^{2 \pi i / n}$ and $\gamma := e^{2 \pi i / r}$. 
If $-1 \in K$, then multiplication by $-1$ permutes the elements of $K$, so we have 
\begin{align} 
c_{x^d}^* = \left(\frac{1}{m} \sum_{k \in K} \omega^{-x^d k}\right) = c_{x^d} \\ 
\end{align} 
It follows that $c_{x^d}$ is real. Furthermore, in this case we have 

\begin{align} 
w_i^* & = \sum_{j = 1}^r \gamma^{-(i-1)(j-1)} c_{x^{j-1}}^* \\ 
& = \sum_{j = 1}^r \gamma^{((r - i+2) - 1)(j-1)} c_{x^{j-1}} \\ 
& = w_{r - i + 2}.  
\end{align} 

If instead $-1 \in x^{\frac{r}{2}} K$, multiplication by $-x^\frac{r}{2}$ permutes the elements of $K$, so 
\begin{align} 
c_{x^d} & = \frac{1}{m} \sum_{k \in K} \omega^{- x^d x^{\frac{r}{2}} k} = c_{x^{d + r/2}}^*. 
\end{align} 

In this case, 
\begin{align} 
w_i^* & = \sum_{j = 1}^r \gamma^{-(i-1)(j-1)} c_{x^{j-1}}^* \\ 
& = \gamma^{(i-1) \frac{r}{2}} \sum_{j = 1}^r \gamma^{(r - i + 1)(j-1 + \frac{r}{2})} c_{x^{j-1 + \frac{r}{2}}} \\ 
& = (-1)^{i-1} w_{r - i + 2}. 
\end{align} 
\end{proof}

We are now ready to prove our theorem: 

\begin{proof} 
\textit{(Theorem \ref{thm:coherenceupperboundmodd})} Since both $n$ and $m$ are odd, then from Lemma \ref{lem:mevennegoneinA} we know that $r$ is even and $-1 \in x^{\frac{r}{2}} K$. As before, set $\omega = e^{2 \pi i/n}$,  $c_{x^d} = \frac{1}{m} \sum_{k \in K} \omega^{x^d k}$, and $\cc = [c_1, c_x, c_{x^2}, ..., c_{x^{r-1}}]^T$. Let $\gamma = e^{2 \pi i / r}$, $F$ the $r \times r$ Fourier matrix with entries $F_{ij} = \gamma^{(i-1)(j-1)}$, and $\w = [w_1, ..., w_r]^T = F \cc$. 

From an inverse Fourier transform, we have $c_{x^{i-1}} = \frac{1}{r} \sum_{j = 1}^{r} \gamma^{-(i-1)(j-1)} w_j$, and in light of Lemmas \ref{lem:wvec} and \ref{lem:cstarwstar}, we can write this as 
\begin{small}
\begin{align} 
\nonumber c_{x^{i-1}} & = \frac{1}{r} \bigg[ w_1 + \gamma^{-(i-1)\frac{r}{2}} w_{\frac{r}{2} + 1}  \\ 
 & \hspace{5pt} + \sum_{j = 2}^{\frac{r}{2}} \left( \gamma^{-(i-1)(j-1)} w_j + \gamma^{-(i-1)((r - j + 2) -1)} w_{r - j + 2} \right) \bigg] \label{eqn:cwsummation0} \\ 
\nonumber & = \frac{1}{r} \bigg[ - \frac{1}{m} + (-1)^{i-1} w_{\frac{r}{2} + 1}  \\ 
& \hspace{5pt} + \sum_{j = 2}^{\frac{r}{2}} \left( \gamma^{-(i-1)(j-1)} w_j + (-1)^{j-1} \left(\gamma^{-(i-1)(j-1)} w_j \right)^* \right) \bigg]. \label{eqn:cwsummation1}
\end{align} 
\end{small}

From Lemma \ref{lem:wvec}, we may write $\gamma^{-(i-1)(j-1)} w_j = \beta e^{i \theta_j}$ for each $j = 2, ..., \frac{r}{2}$, where $\beta = \sqrt{\frac{1}{m} \left( r + \frac{1}{m} \right) }$. Then (\ref{eqn:cwsummation1}) becomes 
\begin{small} 
\begin{align} 
\frac{1}{r} \bigg[ - \frac{1}{m} + (-1)^{i-1} w_{\frac{r}{2} + 1}  + \sum_{j~ \text{even}} 2 i \beta \sin(\theta_j)+ \sum_{j~ \text{odd}} 2 \beta \cos(\theta_j) \bigg]. \label{eqn:cwsummation2}
\end{align} 
\end{small} 

Let $n_e$ be the number of even integers $j$ in the set $\{2, ..., \frac{r}{2}\}$, and $n_o$ the number of odd such integers. Lemma \ref{lem:wvec} tells us that $|w_{\frac{r}{2} + 1}| = \beta$. From Lemma \ref{lem:cstarwstar} we know that when $\frac{r}{2}$ is even $w_{\frac{r}{2} + 1}$ is purely real, and when $\frac{r}{2}$ is odd it is imaginary. In the former case, we may upper-bound $|c_{x^{i-1}}|$ by 
\begin{small} 
\begin{align} 
\frac{1}{r} \left| \frac{1}{m} + \beta + n_e 2i \beta + n_o 2 \beta \right| & = \frac{1}{r} \sqrt{ \left(\frac{1}{m} + \beta (1 +  2 n_o) \right)^2  + (2 n_e \beta)^2 }, \label{eqn:cohboundrdiv2even_2} 
\end{align} 
\end{small} 
and in the latter case we obtain the upper bound 
\begin{small} 
\begin{align} 
\frac{1}{r} \left| \frac{1}{m} + i \beta +  n_e 2i \beta + n_o 2 \beta \right| = \frac{1}{r} \sqrt{ \left(\frac{1}{m} + 2 n_o \beta \right)^2  + \beta^2 (1 + 2 n_e)^2 }. \label{eqn:cohboundrdiv2odd_2}
\end{align} 
\end{small} 

It is not too hard to see that when $\frac{r}{2}$ is even, we have $n_e = \frac{r}{4}$ and $n_o = \frac{r}{4} - 1$. When $\frac{r}{2}$ is odd we obtain $n_e = n_o = \frac{r}{4} - \frac{1}{2}$. We can now substitute these sets of values into (\ref{eqn:cohboundrdiv2even_2}) and (\ref{eqn:cohboundrdiv2odd_2}) respectively, and in either case we obtain the upper bound 
\begin{align} 
|c_{x^{i-1}}| & \le \frac{1}{r} \sqrt{\left(\frac{1}{m} + \left(\frac{r}{2} - 1 \right) \beta \right)^2 + \left(\frac{r}{2}\right)^2 \beta^2}.
\end{align} 
Since our coherence is the maximum of the $|c_{x^{i-1}}|$, we are done. 
\end{proof}

\end{appendices}

\begin{IEEEbiography}[{\includegraphics[width=1in,height=1.25in,clip,keepaspectratio]{Mattface}}]{Matthew Thill}
was born in Arlington Heights, IL. He received the B.S. degree in mathematics and the M.S. degree in electrical engineering in 2009 and 2012 respectively, both from the California Institute of Technology. He is currently completing his graduate studies in electrical engineering, also at Caltech, with research interests in communications, signal processing, coding theory, and network information theory. He was an NDSEG fellow from 2011 to 2014.
\end{IEEEbiography}

\begin{IEEEbiography}[{\includegraphics[width=1in,height=1.25in,clip,keepaspectratio]{hassibi-b120}}]{Babak Hassibi}
was born in Tehran, Iran, in 1967. He received the B.S. degree from the University of Tehran in 1989, and the M.S. and Ph.D. degrees from Stanford University in 1993 and 1996, respectively, all in electrical engineering. 

He has been with the California Institute of Technology since January 2001, where he is currently the Gordon M Binder/Amgen Professor Of Electrical Engineering. From 2008-2015 he was Executive Officer of Electrical Engineering, as well as Associate Director of Information Science and Technology. From October 1996 to October 1998 he was a research associate at the Information Systems Laboratory, Stanford University, and from November 1998 to December 2000 he was a Member of the Technical Staff in the Mathematical Sciences Research Center at Bell Laboratories, Murray Hill, NJ. He has also held short-term appointments at Ricoh California Research Center, the Indian Institute of Science, and Linkoping University, Sweden. His research interests include wireless communications and networks, robust estimation and control, adaptive signal processing and linear algebra. He is the coauthor of the books (both with A.H.~Sayed and T.~Kailath) {\em Indefinite Quadratic Estimation and Control: A Unified Approach to H$^2$ and H$^{\infty}$ Theories} (New York: SIAM, 1999) and {\em Linear Estimation} (Englewood Cliffs, NJ: Prentice Hall, 2000). He is a recipient of an Alborz Foundation Fellowship, the 1999 O. Hugo Schuck best paper award of the American Automatic Control Council (with H.~Hindi and S.P.~Boyd), the 2002 National ScienceFoundation Career Award, the 2002 Okawa Foundation Research Grant for Information and Telecommunications, the 2003 David and Lucille Packard Fellowship for Science and Engineering,  the 2003 Presidential Early Career Award for Scientists and Engineers (PECASE), and the 2009 Al-Marai Award for Innovative Research in Communications, and was a participant in the 2004 National Academy of Engineering ``Frontiers in Engineering'' program. 

He has been a Guest Editor for the IEEE Transactions on Information Theory special issue on ``space-time transmission, reception, coding and signal processing'' was an Associate Editor for Communications of the IEEE Transactions on Information Theory during 2004-2006, and is currently an Editor for the Journal ``Foundations and Trends in Information and Communication'' and for the IEEE Transactions on Network Science and Engineering. He is an IEEE Information Theory Society Distinguished Lecturer for 2016-2017.
\end{IEEEbiography} 


\begin{thebibliography}{1} 
\footnotesize{

\bibitem{BajwaCalderbankJafarpour} W. U. Bajwa, R. Calderbank, S. Jafarpour, ``Why Gabor frames? Two fundamental measures of coherence and their role in model selection,'' \textit{J. Commun. Netw.}, 12: 289-307, 2010. 

\bibitem{Bandyopadhyay} S. Bandyopadhyay, P. O. Boykin, V. Roychowdhury, F. Vatan, ``A new proof of the existence of mutually unbiased bases. \textit{Algorithmica}, 34:512-528, 2002. 

\bibitem{Baumert} L. D. Baumert, in \textit{Cyclic Difference Sets, Lecture Notes in Mathematics}, Berlin, Germany, 1971, vol. 182, Springer. 

\bibitem{BenedettoFickus} J. Benedetto and M. Fickus, ``Finite Normalized Tight Frames,'' \textit{Advances in Computational Mathematics}, 18 (2-4), 357-385, 2003. 

\bibitem{BethJungnickel} T. Beth, D. Jungnickel, and H. Lenz, \textit{Design Theory}, Cambridge, U.K.: Cambridge University Press, 1999. 

\bibitem{BosWaldron} L. Bos and S. Waldron, ``Some remarks on Heisenberg frames and sets of equiangular lines,'' \textit{N. Z. Jour. Math}. 36, 113-137, 2007. 

\bibitem{Tao} E. J. Candes and T. Tao, ``Decoding by linear programming,'' {\it IEEE Trans. Inform. Th.} 51, 4203-4215, 2005. 

\bibitem{Candes} E. J. Candes, ``The restricted isometry property and its implications in compressed sensing,'' \textit{C. R. Acad. Sci. Paris S'er. I Math.} 346, 589-592, 2008. 

\bibitem{Romberg} E. J. Candes, J. Romberg and T. Tao, ``Stable signal recovery from incomplete and inaccurate measurements,'' \textit{Comm. Pure Appl. Math.} 59, 1208-1223, 2006. 

\bibitem{CasazzaArt} P. G. Casazza, ``The Art of Frame Theory,'' \textit{Taiwanese Journal of Mathematics}, Vol. 4, No. 2, pp. 129-201, June 2000. 

\bibitem{CasazzaKutyniokFickusMixon} P. G. Casazza, M. Fickus, D. G. Mixon, Y. Wang, Z. Zhou, ``Constructing tight fusion frames,'' \textit{Applied and Computational Harmonic Analysis} 30 (2), 175-187, 2011. 

\bibitem{Kutyniok} P.G. Casazza, G. Kutyniok, and M.C. Lammers, ``Duality principles in frame theory.'' \textit{J. Fourier Anal. Appl.} 10, 383-408, 2004. 

\bibitem{CasazzaKovacevic} P. G. Casazza and J. Kova\u{c}evi\'{c}, ``Equal-norm tight frames with erasures,'' \textit{Advances in Computational Mathematics}, Springer, 2003. 

\bibitem{CasazzaKutyniok} P.G. Casazza and G. Kutyniok, ``Finite Frames: Theory and Applications,'' Chapter 5 (written by S. Waldron). Springer, 2012. 

\bibitem{CasazzaKutyniok2} P. G. Casazza and G. Kutyniok, ``Frames of Subspaces,'' \textit{Wavelets, Frames and Operator Theory}, Contemp. Math., College Park, MD, 2003, vol. 345, Amer. Math. Soc., Providence, RI pp. 87-113, 2004. 

\bibitem{CasazzaKutyniokLi} P. G. Casazza, G. Kutyniok, and S. Li, ``Fusion frames and distributed processing,'' \textit{Applied and Computational Harmonic Analysis}, Vol. 25, Iss. 1,  pp. 114-132, July 2008. 

\bibitem{ChienWaldron} T. Chien, S. Waldron, ``A classification of the harmonic frames up to unitary equivalence,'' \textit{Appl. Comput. Harmon. Anal}. 30, 307-318, 2011. 


\bibitem{Chu} D. Chu, ``Polyphase codes with good periodic correlation properties (Corresp.),'' \textit{IEEE Trans. Inform. Theory}, vol. 18, no. 4, pp. 531-532, July 1972. 


\bibitem{ConwayHarding} J. H. Conway, R. H. Harding, and N. J. A. Sloane, ``Packing lines, planes, etc.: Packings in Grassmannian spaces,'' \textit{Exp. Math.}, vol. 5, no. 2, pp. 139-159, 1996.

\bibitem{Davenport} M. A. Davenport and M. B. Wakin, ``Analysis of orthogonal matching pursuit Using the restricted isometry property," {\it IEEE Trans. Inform. Theory}, 56(9):4395-4401, 2010.

\bibitem{Delsarte} P. Delsarte, J. M. Goethals and J. J. Seidel, ``Spherical Codes and Designs,'' \textit{Geometriae Dedicata}, Vol. 6, No. 3, 363-388, 1977. 

\bibitem{Ding} C. Ding, ``Complex codebooks from combinatorial designs,'' \textit{IEEE Trans. Inf. Theory}, vol. 52, no. 9, pp. 4229-4235, Sep. 2006. 

\bibitem{DingFeng} C. Ding and T. Feng, ``A Generic Construction of Complex Codebooks Meeting the Welch Bound,'' \textit{IEEE Trans. Inf. Theory}, Vol. 53 ,  Issue 11, 2007. 

\bibitem{DingGolinKlove} C. Ding, M. Golin, and T. Kl\o ve, ``Meeting the Welch and Karystinos-Pados bounds on DS-CDMA binary signature sets,'' \textit{Designs, Codes, Cryptogr.}, vol. 30, pp. 73-84, 2003. 

\bibitem{Donoho} D. L. Donoho and M. Elad, ``Optimally sparse representations in general (non-orthogonal) dictionaries via $\ell_1$ minimization,'' \textit{Proc. Nat. Acad. Sci.} 100, 2197-2202, 2002. 

\bibitem{Huo} D. L. Donoho and X. Huo, ``Uncertainty principles and ideal atomic decompositions,'' {\it IEEE Trans. Inform. Theory} 47, 2845-2862, 2001. 

\bibitem{Eldar} T. C. Eldar and G. D. Forney, Jr., ``Optimal tight frames and quantum measurement,'' \textit{IEEE Trans. Inform. Theory}, 48, 599-610, 2002. 

\bibitem{Fan} J. L. Fan, ``Array codes as Low-Density Parity Check codes,'' \textit{Proc. Int'l. Symp. on Turbo Codes}, Brest, France, 543-546, Sept. 2000. 

\bibitem{FickusMixonTremain} M. Fickus, D. G. Mixon, J. C. Tremain, ``Steiner equiangular tight frames,'' \textit{Linear Algebra and its Applications}, 436 (5), 1014-1027, 2012. 

\bibitem{Gabardo} J. P. Gabardo and D. Han, ``Frame representations for group-like unitary operator systems. \textit{Journal of Operator Theory}, 49, 223-244, 2003. 

\bibitem{Golomb} S. W. Golomb, ``Cyclic Hadamard difference sets---Constructions and applications,'' \textit{Sequences and their Applications}, pp. 39-48, Springer, London, U.K., 1999.  

\bibitem{HayWaldron} N. Hay, S. Waldron, ``On computing all harmonic frames of $n$ vectors in $\C^d$,'' Appl. Comput. Harmon. Anal. 21, 168-181, 2006. 

\bibitem{Pailraj} R. W. Heath, Jr. and A. J.  Paulraj, ``Linear dispersion codes for MIMO systems based on frame theory,'' \textit{IEEE Trans. on Sig. Proc}. Vol. 50, Issue 10, 2429-2441, 2002. 

\bibitem{Bolcskei} R. W. Heath, Jr., H. Bolcskei and A. J. Paulraj, ``Space-time signaling and frame theory,'' \textit{IEEE Proceedings of ICASSP}, 4, 2445-2448, 2001. 

\bibitem{Ipatov} V. P. Ipatov, ``On the Karystinos-Pados bounds and optimal binary DS-CDMA signature ensembles,'' \textit{IEEE Commun. Lett.}, vol. 8, no. 2, pp. 81-83, Feb. 2004. 

\bibitem{Jones} H. F. Jones, ``Groups, Representations, and Physics Second Edition,'' Institute of Physics Publishing, Bristol and Philadelphia, 1998. 

\bibitem{Kahn} F. Kahn, ``LTE for 4G Mobile Broadband,'' Cambridge University Press, New York, 2009. 

\bibitem{Kaira} D. Kalra, ``Complex equiangular cyclic frames and erasures,'' \textit{Linear Algebra Appl.}, 419, 373- 399, 2006. 

\bibitem{KarystinosPados} G. N. Karystinos and D. A. Pados, ``New bounds on the total-squared-correlation and perfect design of DS-CDMA binary signature sets,'' \textit{IEEE Trans. Commun.}, vol. 3, pp. 260-265, 2003.

\bibitem{Khatirinejad} M. Khatirinejad, ``On Weyl-Heisenberg orbits of equiangular lines,'' J. Algebr. Comb. 28, 333- 349, 2008. 

\bibitem{Klappenecher} A. Klappenecher, M. R\"{o}tteler, ``Constructions of Mutually Unbiased Bases,'' \textit{Finite Fields and Applications}, 7th International Conference, Fq7, Toulouse, France, May 5-9, 2003. Revised Papers. 

\bibitem{Kovacevic} J. Kovacevic, A. Chebira, ``Life beyond bases: The advent of frames,'' \textit{IEEE Signal Processing Magazine}, 2007. 

\bibitem{MixonBajwaCalderbank} D. Mixon, W. U. Bajwa, R. Calderbank, ``Frame Coherence and Sparse Signal Processing,'' \textit{IEEE Proc. of ISIT}, 2011. 

\bibitem{RenesBlumeKohoutScottCaves} J. M. Renes, R. Blume-Kohout, A. J. Scott, C. M. Caves, ``Symmetric informationally complete quantum measurements,'' \textit{J. Math. Phys}., 45, 2171-2180, 2004. 

\bibitem{ScottGrassi} A. J. Scott, M. Grassl, ``SIC-POVMs: A new computer study,'' arXiv:0910.5784v2 [quant-ph], 2009. 

\bibitem{Scott} A. J. Scott, ``Tight informationally complete quantum measurements," \textit{J. Phys. A: Math. Gen.} \textbf{39} 13507, 2006. 

\bibitem{Hassibi} A. Shokrollahi, B. Hassibi, B.M. Hochwald and W. Sweldens, Representation theory for high-rate multiple-antenna code design, IEEE Transactions on Information Theory, vol.47, no.6, pages 2335-67, Sept. 2001. 

\bibitem{Slepian} D. Slepian, ``Group Codes for the Gaussian Channel,'' \textit{Bell Sys. Tech. J.}, vol. 47, pp. 575-602, Apr. 1968. 

\bibitem{StrohmerGabor} T. Strohmer. Approximation of dual Gabor frames, window decay, and wireless communications. {\it Appl. Comp. Harm. Anal.}, 11(2):243-262, 2001. 

\bibitem{Heath} T. Strohmer and R. W. Heath Jr., ``Grassmannian frames with appications to coding and communication,'' {\it Appl. Comput. Harmon. Anal.} 14, 257-275, 2003. 

\bibitem{Sustik} M. Sustik, J. A. Tropp, I. S. Dhillon, and R. W. Heath, ``On the existence of equiangular tight frames,'' \textit{Linear Alg. and Applications}, 426 No. 2-3, 619-635, 2007. 

\bibitem{ThillGroupFourierFrames} M. Thill and B. Hassibi, ``Low-Coherence Frames from Group Fourier Matrices,'' Under Preparation. 

\bibitem{ThillAllerton} M. Thill and B. Hassibi, ``Frames, Group Codes, and Subgroups of $(\mathbb{Z}/p\mathbb{Z})^\times$,'' \textit{Proc. of Allerton}, 2012. 

\bibitem{Thill_ICASSP2013} M. Thill and B. Hassibi, ``Frames from Groups: Generalized Bounds and Dihedral Groups,'' \textit{ICASSP}, 2013. 

\bibitem{Thill_ISIT2013} M. Thill and B. Hassibi, ``On frames from abelian group codes,'' \textit{Proceedings of ISIT}, 2013. 

\bibitem{Tropp07} J. Tropp and A. Gilbert, ``Signal recovery from partial information via orthogonal matching pursuit,'' \textit{IEEE Trans. Inform. Theory}, vol. 53, no. 12, pp. 4655-4666, 2007. 

\bibitem{Tropp05} J.A. Tropp, I.S. Dhillon, R. Heath Jr., and T. Strohmer, ``Structured Tight Frames via an Alternating Projection Method,'' \textit{IEEE Trans. Inform. Theory}, vol.51(1):188-209, 2005. 

\bibitem{ValeWaldron3} R. Vale, S. Waldron, ``The symmetry group of a finite frame,'' \textit{Linear Algebra Appl.}, 433, 248- 262, 2010. 

\bibitem{Vale} R. Vale and S. Waldron, ``Tight frames and their symmetries,'' \textit{Constr. Approx.} 21, 83-112, 2005. 

\bibitem{ValeWaldron2} R. Vale, S. Waldron, ``Tight frames generated by finite nonabelian groups,'' \textit{Numer. Algorithms}, 48, 11-27, 2008. 

\bibitem{VasicMilenkovic} B. V. Vasic and O. Milenkovic, ``Combinatorial Constructions of Low-Density Parity-Check Codes for Iterative Decoding,'' \textit{IEEE Trans. Inform. Theory}, 50(6): 1156-1176, 2004. 

\bibitem{Waldron1} S. Waldron, ``An Introduction to Finite Tight Frames,'' Springer, NewYork, 2011. 

\bibitem{Welch} L. Welch, ``Lower bounds on the maximum cross correlation of signals,'' \textit{IEEE Trans. Inform. Theory}, vol. 20, no. 3, pp. 397-399, May 1974. 

\bibitem{Wootters} W. K. Wootters, B. D. Fields, ``Optimal state-determination by mutually unbiased measurements,'' \textit{Ann. Physics}, 191:363-381, 1989. 

\bibitem{Giannakis} P. Xia, S. Zhou, G. B. Giannakis, ``Achieving the Welch bound with difference sets,'' {\it IEEE Trans. Inform. Theory} 51, 1900-1907, 2005. 

}
\end{thebibliography}
\end{document}